\documentclass[11pt, letterpaper, onecolumn]{IEEEtran}

\IEEEoverridecommandlockouts
\usepackage[margin=0.70 in]{geometry}
\usepackage{amsmath,amsfonts,amsthm,enumerate}
\usepackage{amssymb}
\usepackage{graphicx}
\usepackage{cite,bm}
\usepackage{color,epsfig}
\usepackage{authblk}

\theoremstyle{plain}
\newtheorem{mythm}{Theorem}
\newtheorem{mylemma}{Lemma}

\theoremstyle{definition}
\newtheorem{mydef}{Definition}
\newtheorem{myexample}{Example}

\theoremstyle{remark}
\newtheorem{myrm}{Remark}

\usepackage[ colorlinks = true,
             linkcolor = blue,
             urlcolor  = blue,
             citecolor = red,
             anchorcolor = green,
]{hyperref}

\allowdisplaybreaks[1]

\begin{document}
\title{A Case Where Interference Does Not Affect The Channel Dispersion} 

\author{Sy-Quoc Le,$\,\,\,$ Vincent Y.~F.\ Tan,$\,\,\,$ Mehul Motani \thanks{The authors are   with the Department of Electrical and Computer Engineering (ECE), National University of Singapore (NUS). V.~Y.~F.\ Tan is also with the Department of Mathematics, NUS.  The authors' emails are \url{le.sy.quoc@nus.edu.sg},  \url{vtan@nus.edu.sg} and  \url{motani@nus.edu.sg}. } } 
 
\maketitle
\begin{abstract}

In 1975, Carleial presented a special case of an interference channel in which the interference does not reduce   the capacity of the constituent point-to-point Gaussian channels. In this work, we show that if the inequalities in the conditions that Carleial stated are strict, the dispersions are similarly unaffected. More precisely, in this work, we characterize the second-order coding rates of the Gaussian interference channel in the strictly very strong interference regime. In other words, we characterize the speed of convergence of rates of optimal block codes towards a boundary point of the (rectangular) capacity region. These second-order rates are expressed in terms of the   average probability of error and variances of some  modified information densities which coincide with the dispersion of the (single-user) Gaussian channel. We thus conclude that the dispersions are unaffected by interference in  this channel model. 
\end{abstract}

\section{Introduction}
Recently,  the study  of second-order coding rates  for fixed error probabilities has become an increasingly prominent research topic in network information theory because the analysis provides key insights into the  (delay-constrained) performance of the communication systems in the finite blocklength regime~\cite{Polyanskiy2010}. Strassen \cite{Strassen62}, Hayashi \cite{Hayashi2009}, and Polyanskiy, Poor and Verd\'u~\cite{Polyanskiy2010} characterized the second-order coding rate of the discrete memoryless (DM) point-to-point channel and the additive white Gaussian noise (AWGN) point-to-point channel. The result can be summarized as follows. If $M^*(n,\epsilon, \mathsf{SNR})$ denotes the maximum number of codewords that can be transmitted over $n$ uses of a discrete-time AWGN channel with signal-to-noise ratio $\mathsf{SNR}$ and average error probability no larger than $\epsilon \in (0,1)$, then,  it was shown by \cite{Polyanskiy2010} and \cite{TanTom13} that 
\begin{equation}
\log M^*(n,\epsilon,\mathsf{SNR}) = n\mathsf{C}(\mathsf{SNR}) + \sqrt{n \mathsf{V} (\mathsf{SNR})}\Phi^{-1}(\epsilon) + \frac{1}{2} \log n + O(1) \label{eqn:gauss1}  
\end{equation}
where $\Phi(\cdot)$ is the cumulative distribution function of the standard Gaussian, and the {\em Gaussian capacity} $\mathsf{C}(\mathsf{SNR})$  and {\em Gaussian dispersion} functions $\mathsf{V}(\mathsf{SNR})$ are respectively defined as 
\begin{align}
\mathsf{C}(\mathsf{SNR}) &\triangleq \frac{1}{2}\log (1+\mathsf{SNR}) \mbox{ nats per channel use} ,  \label{eqn:g_cap}
\end{align}
and
\begin{align}
\mathsf{V}(\mathsf{SNR}) &\triangleq \frac{\mathsf{SNR}(\mathsf{SNR}+2)}{2(\mathsf{SNR}+1)^2} \mbox{ nats$^2$ per channel use}. \label{eqn:disp_def}
\end{align}

The sum of the first two terms of equation (1), namely $n\mathsf{C}(\mathsf{SNR}) + \sqrt{n \mathsf{V}(\mathsf{SNR})} \Phi^{-1}(\epsilon)$, is called the {\em normal approximation} to the logarithm of the size of the optimal codebooks $\log M^*(n,\epsilon,\mathsf{SNR})$. Since it has been shown that the normal approximation is a good proxy to the finite blocklength fundamental limits \cite{Polyanskiy2010} at moderate blocklengths, the result can be interpreted as follows: If a system designer desires to use a Gaussian communication channel up to $n$ times with a tolerable average error probability not exceeding $\epsilon$, the maximum number of nats of information he can communicate is roughly $n \mathsf{C}(\mathsf{SNR}) + \sqrt{n \mathsf{V}(\mathsf{SNR})} \Phi^{-1}(\epsilon)$. Thus, for $\epsilon<0.5$, the backoff from the Shannon limit (Gaussian capacity) is  $\sqrt{ {\mathsf{V}(\mathsf{SNR})}/{n}} \, \Phi^{-1}(1-\epsilon)$ (a positive quantity). The constraint on the blocklength is motivated by real-world, delay-constrained applications such as real-time multimedia streaming. In such applications, the communication data is usually divided into a stream of packets, which have to arrive at their desired destinations within a certain acceptable, and usually short, delay.

The quantities $\mathsf{C}(\mathsf{SNR})$ and $\mathsf{V}(\mathsf{SNR})$ are respectively the expectation and the conditional variance of an appropriately defined information density random variable. These are information-theoretic quantities that characterize the information transmission capability of the channel. In fact, $\mathsf{V}(\mathsf{SNR})$, coined the ``dispersion'' by Polyanskiy-Poor-Verd\'u~\cite{Polyanskiy2010}, is a channel-dependent quantity that characterizes the speed at which the rates of capacity-achieving codes converge to the Shannon limit. The {\em second-order coding rate}, a term coined by Hayashi~\cite{Hayashi08,Hayashi2009}, is a different, but related, object. It is  the coefficient of the $\sqrt{n}$ term in \eqref{eqn:gauss1}, namely $\sqrt{\mathsf{V}(\mathsf{SNR})}\Phi^{-1}(\epsilon)$. More precisely, the $(\kappa,\epsilon)$-second-order coding rate $L^*(\kappa,\epsilon)\in\mathbb{R}$ is the maximum $L$ for which there exists a sequence of length-$n$ block codes of sizes $M_n$ and error probabilities asymptotically not exceeding $\epsilon$ such that

\begin{equation}
\log M_n \ge n\kappa +\sqrt{n} L + o\big(\sqrt{n}\big). \label{EquationL}
\end{equation}

If $\kappa<\mathsf{C}(\mathsf{SNR})$, then it can be seen by the direct part of the coding theorem for the AWGN channel that $L^*(\kappa,\epsilon)=\infty$. If the {\em strong converse} holds (and for the AWGN channel it does~\cite{Yoshihara}), then for all $\kappa> \mathsf{C}(\mathsf{SNR})$, the $(\kappa,\epsilon)$-second-order coding rate $L^*(\kappa,\epsilon)=-\infty$. Hence, the only non-trivial case is the phase-transition point $\kappa= \mathsf{C}(\mathsf{SNR})$. Hayashi's result is that~\cite{Hayashi2009}

\begin{equation}
L^*(\mathsf{C}(\mathsf{SNR}),\epsilon)= \sqrt{\mathsf{V}(\mathsf{SNR})}\Phi^{-1}(\epsilon),
\end{equation}
which implies the set of real numbers  $L$ satisfying
\begin{equation}
L \leq  \sqrt{\mathsf{V}(\mathsf{SNR})}\Phi^{-1}(\epsilon),
\end{equation}
is second-order achievable, i.e., there exists a sequence of length-$n$ block codes, with  average error probabilities not exceeding $\epsilon$ asymptotically, and  fixed sizes $M_n$, such that (\ref{EquationL}) holds.

Note that second-order coding rates can be negative depending on $\epsilon$. 
Since the problem we are solving in this paper is a multi-terminal one, we focus on characterization of the {\em set of} achievable second-order coding rates $(L_1,L_2)$, which is a subset of the real plane.

\begin{figure}[t]
\centering
\setlength{\unitlength}{.05cm}
\begin{picture}(125, 82)
\put(0, 5){\vector(1, 0){120}}
\put(5, 0){\vector(0, 1){80}}

\put(120, 7){\mbox{$R_1$}}
\put(7, 82){\mbox{$R_2$}}

\put(-1, -3){\mbox{$0$}}

\put(72, -3){\mbox{$I_{11}$}}
\put(6, 68){\mbox{$I_{21}$}}

\put(80,50){\circle*{4}}
\put(83, 50){(i)}
\put(80,65){\circle*{4}}
\put(82, 68){(ii)}
\put(65,65){\circle*{4}}
\put(56,69){(iii)}

\linethickness{.02in}
\put(80, 5){\line(0, 1){60}}
\put(5, 65){\line(1, 0){75}}
\end{picture}
\caption{Illustration of the capacity region of the Gaussian IC with very strong interference \cite{Carleial75}. The signal-to-noise ratios $S_j = h_{jj}^2 P_j$  and $I_{11}=\mathsf{C}(S_1)$ and $I_{21}=\mathsf{C}(S_2)$. }
\label{fig:wz}
\end{figure}
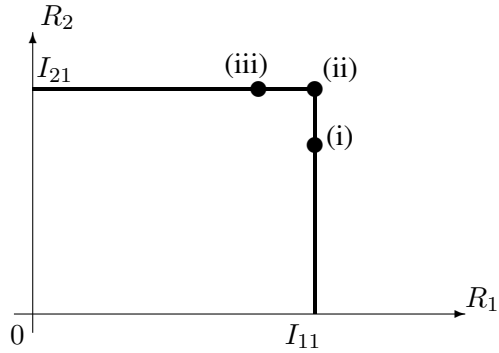

\subsection{Prior Work}
Following the pioneering works in \cite{Strassen62,Polyanskiy2010,Hayashi2009}, there have been many follow-up works for various point-to-point models~\cite{Hoydis12, Yang13, Ingber2010,TT2013}, for source coding \cite{KV12, IK11, KontoVerdu12, Kosut13}, for joint source-channel coding \cite{WangAllerton11, KostinaIT13}, and for coding with side-information \cite{WKT13}. However, it is not trivial to  generalize these results from  the  single- to  the multi-user setting. Thus far, there have been only a few second-order works  for multi-user settings. Hence, the understanding  is far from being complete. Initial efforts focused on {\em global achievable dispersions}\cite{Haim12} for the DM multiple-access channel (MAC) \cite{TanKosut12, Molavian12, Huang12, ML13}, for the DM asymmetric broadcast channel \cite{TanKosut12}, and for the DM interference channel (IC) \cite{LTM13}. However, as pointed out by Haim et al.~\cite{Haim12}, global dispersion analysis has certain drawbacks such as the failure to precisely capture the nature of convergence to the boundary of the capacity region, the inability  in characterizing the deviation from a specific point on the boundary and the difficulty in obtaining conclusive second-order results. To overcome these weaknesses, Haim et al.~\cite{Haim12} proposed {\em local  dispersion analysis}. Tan-Kosut \cite{TanKosut12} and Nomura-Han \cite{Nomura13} characterized the second-order optimal rate region (the set of achievable second-order coding rates for fixed error probability $\epsilon$ and a fixed point on the optimal rate region) for  distributed source coding, i.e., the Slepian-Wolf problem~\cite{sw73}. While it is possible to obtain tight second-order converse bounds for distributed source coding, it is challenging to do similarly for channel coding problems such as the DM-MAC. This is due in part to the union over independent input distributions. Scarlett-Tan \cite{ST13} recently obtained the second-order capacity region for the Gaussian MAC with degraded message sets. The degradedness of the message sets makes it possible to avoid certain difficulties to get a tight converse by appealing to the reductions similar to the method of types. The local second-order capacity region for the Gaussian MAC with non-degraded message sets is an open problem.

\subsection{Our Main Contribution}
In this paper, we study the local dispersions of the Gaussian IC in the strictly very strong  interference regime. Carleial showed that the capacity region of the very strong Gaussian IC (which includes the strictly very strong Gaussian IC) is a rectangle \cite{Carleial75}, as shown in Figure~\ref{fig:wz}. We characterize the so-called {\em second-order capacity region}, which we   briefly explain here. We   fix a point $(\kappa_1,\kappa_2)$ lying on the boundary of the   capacity region. We also fix an admissible error probability  $\epsilon\in (0,1)$. We then characterize the set of pairs $(L_1, L_2)$ for which there exists a sequence of blocklength-$n$ codes with $M_{jn}$   codewords,  and average error probabilities not exceeding $\epsilon$ asymptotically,  such that 
\begin{equation}
\log M_{jn}\ge n\kappa_j + \sqrt{n } L_j + o(\sqrt{n}), \label{eqn:num_cws}
\end{equation}
for $j =1,2$. The converse is proved using a generalized version of the Verd\'u-Han lemma~\cite{VH94}, which involves only two error events. The direct part is proved using a generalized version of  Feinstein's lemma~\cite{Feinstein}, which involves four error events. The condition of being in the strictly very strong interference regime reduces the number of error events involved in the direct part, thus allowing the converse to match the direct part.  Our key contribution is the  determination of  the set of second-order rate pairs $(L_1,L_2)$, which characterize the rate of convergence of optimal (first-order) rates  to a particular point $(\kappa_1,\kappa_2)$ lying on the boundary of the capacity region. One of the interesting observations is that,  if $(\kappa_1,\kappa_2)$ is the corner point of the rectangular capacity region (case (ii) in Figure~\ref{fig:wz}), then the set of all such $(L_1, L_2) \in \mathbb{R}^2$ is given by 
\begin{equation}
\Phi\left( - \frac{L_1}{\sqrt{V_1}} \right)\Phi\left( - \frac{L_2}{\sqrt{V_2}} \right)\ge 1-\epsilon, \label{eqn:main_res_intro}
\end{equation}
where $V_j \triangleq \mathsf{V}(\mathsf{SNR}_j)$ is the effective  Gaussian dispersion of the channel from the $j^{\mathrm{th}}$ transmitter to the $j^{\mathrm{th}}$ receiver, i.e., $V_j$ is equal to~\eqref{eqn:disp_def}  evaluated at signal-to-noise ratio $\mathsf{SNR}_j$. An illustration of the $(L_1, L_2)$ region is provided in Figure~\ref{fig:ex}. We see  from \eqref{eqn:main_res_intro} that the two channels appear to operate independently of each other. Indeed $\Phi\left( - {L_j}/ {\sqrt{V_j}} \right)$ is asymptotically the probability of correct detection of the $j^{\mathrm{th}}$-channel where the number of codewords for the $j^{\mathrm{th}}$ codebook is given by $M_{jn}$. Intuitively, the inequality in \eqref{eqn:main_res_intro} says that the system does not make an error  if and only if both channels do not err.  Just as Carleial \cite{Carleial75} showed that in the very strong interference regime   the capacities of the constituent channel are not reduced, in the {\em strictly} very strong interference regime, our main result  shows that {\em the dispersions $V_1$ and $V_2$ remain unchanged} and there is no cross-correlation between the two channels in the sense of \eqref{eqn:main_res_intro}. 

We emphasize that apart from Scarlett-Tan's work~\cite{ST13}, this is the only work that completely characterizes the local dispersions for a channel-type network information theory problem. Furthermore, this is the first work which characterizes the local dispersions for a channel-type network information theory problem, where input distributions are of the product form. 

This paper is accepted for and is to be presented in part at International Symposium on Information Theory 2014.

\subsection{Paper Organization}
This paper is organized as follows. The system model is introduced and the problem is formulated in Section \ref{SystemModel}. Next, the main result of the paper is stated and discussed in Section \ref{MainResult}. Future works are then discussed in Section \ref{FutureWork}. All proofs are deferred to the    appendices.

\section{System model and problem formulation} \label{SystemModel}
The two-user Gaussian interference channel (IC) is defined by the following input-output relationships
\begin{align}
Y_{1i} &= h_{11}X_{1i} + h_{21}X_{2i} + Z_{1i}, \\
Y_{2i} &= h_{12}X_{1i} + h_{22}X_{2i} + Z_{2i}, 
\end{align}
where $X_{ji}$ denotes the signal sent by transmitter $j$ ($\mathrm{Tx}_j$ in short), $Y_{j i}$ denotes the output at receiver $j$ ($\mathrm{Rx}_j$ in short), for $j=1,2$, at time $i$, for $i \in \{1,2,...,n\}$, and $\{Z_{j i}\}_{i=1} ^n$  are independent (across time and between users at a fixed time),\footnote{The assumption of independence between the channel noises $Z_{1i}$ and $Z_{2i}$ was not made in Carleial's work~\cite{Carleial75}  (i.e., $Z_{1i}$ and $Z_{2i}$  may be correlated) but we need this assumption for the analyses in the current work. Indeed, this is a common assumption in Gaussian ICs in the literature \cite{HK81}. It is well known that the capacity region of any general  IC depends only on the marginals $W_1$ and $W_2$ \cite[Chapter~6]{elgamal} but it is, in general, not true that   the $(\kappa_1,\kappa_2,\epsilon)$-second-order capacity region (per Definition \ref{def:second_order}) has the same property.} 
additive white  Gaussian noise processes with zero means and unit variances. Denote the input alphabets as $\mathcal{X}_j^n$, and the output alphabets as $\mathcal{Y}_j^n$. Denote the transitional probability 
$ P_{Y_1^n Y_2^n|X_1^n X_2^n} (y_1^n y_2^n|x_1^n x_2^n)$ as  $W^n(y_1^n y_2^n|x_1^n x_2^n)$ for conciseness.
Denote the $Y_1$- and $Y_2$-marginals  of $W$ as $W_1$ and $ W_2$ respectively.                             
The forward channel gains $\{h_{11},h_{21},h_{12},h_{22}\}$  are assumed to be positive  constants and known at all terminals.
Transmitter $\mathrm{Tx}_j$, for $j=1,2$, wishes to communicate a message $S_{j} \in \{ 1,2,...,M_{jn}\}$ to receiver $\mathrm{Rx}_j$.  It is assumed that  the messages $S_1$ and $S_2$ are independent, and uniformly distributed on their respective message sets $\mathcal{W}_j \triangleq \{ 1,2,...,M_{jn}\}$, for  $j = 1,2$.   We use nats as the units of information.

Define  the {\em feasible set}  of channel inputs
\begin{equation}
\mathcal{F}_{jn} \triangleq \left\{x_j^n \in \mathcal{X}_j^n \,\bigg| \, \sum _{k=1} ^{n}  x_{jk}^2\leq nP_j \right\}
\end{equation}
for positive numbers $P_j, j=1,2$. $P_1$ and $P_2$ are the upper bounds on the average powers of the codewords. An {\em $(M_{1n},M_{2n},n, \epsilon_n,P_1,P_2)$-code for the Gaussian IC}  consists  of two encoding functions
$f _{jn} : \mathcal{W}_{j} \rightarrow \mathcal{F}_{jn} $
and two decoding functions
$g_{jn}:  \mathcal{Y}_{j}^n  \rightarrow \hat{\mathcal{W}}_{j} \text{ for } j=1,2,$
where the {\em average probability of error} is defined as
\begin{equation}
\epsilon_n \triangleq \Pr \left(\hat{S}_1 \not= S_1 \text{ or } \hat{S}_2 \not= S_2 \right). \label{eqn:error_prob}
\end{equation}



In the spirit of the works on second-order asymptotics \cite{Hayashi08, Hayashi2009,Nomura13,ST13,TanKosut12}, we define the   second-order capacity region as follows. 
\begin{mydef} \label{def:second_order}
Fix any two non-negative numbers $\kappa_1$ and $\kappa_2$. A   pair $(L_1,L_2)$ is said to be \textit{$(\kappa_1, \kappa_2,\epsilon)$-achievable} \footnote{We note that it is more precise to define a pair being $(P_1,P_2,\kappa_1, \kappa_2, \epsilon)$-achievable. However, we omit the dependence on $(P_1,P_2)$ as $(P_1,P_2)$ are fixed throughout the paper.} if there exists a sequence of  $(M_{1n}, M_{2n},n,\epsilon_n,P_1,P_2)$-codes such that 
\begin{equation}
\limsup_{n \to \infty} \epsilon_n \leq \epsilon,
\end{equation}
and
\begin{equation}
\liminf_{n \to \infty} \frac {1}{\sqrt{n}} (\log M_{jn} - n \kappa_j) \geq L_j  \label{eqn:second_order}
\end{equation}
for $j = 1,2$. The \textit{$(\kappa_1, \kappa_2,\epsilon)$-second-order capacity region} of the IC $\mathcal{L}(\kappa_1,\kappa_2,\epsilon)\subset\mathbb{R}^2$ is defined as the closure of the set of all $(\kappa_1, \kappa_2,\epsilon)$-achievable rate pairs $(L_1,L_2)$.  
\end{mydef}


\begin{mydef} \label{Def2}
The IC is said to have a \textit{very strong interference} if 
\begin{align} 
h_{22}^2 \leq \frac{h_{21}^2}{1+h_{11}^2 P_1}\,\,  \text{ and } \,\, h_{11}^2 \leq \frac{h_{12}^2}{1+h_{22}^2 P_2}. \label{E1}
\end{align}
The IC is said to have a \textit{strictly very strong interference} if both inequalities in (\ref{E1}) are strict.
\end{mydef}

\begin{myexample} \label{eg:first}
Consider a Gaussian IC, where $P_1 = P_2 =1$, $h_{11}= h_{22} =1$, $h_{21} = 3$, and $h_{12}=4$. This is an example of a Gaussian IC in the strictly  very strong interference regime. Clearly, there are uncountably many such examples as long as the interference link gains $h_{21}$ and $h_{12}$ are sufficiently large compared to the direct link gains $h_{11}$ and $h_{22}$  and the admissible powers $P_1$ and~$P_2$.
\end{myexample}

\begin{mydef}
Recall the definition of the Gaussian capacity function $\mathsf{C}(\cdot )$ in \eqref{eqn:g_cap}. Define   the following first-order quantities
\begin{align}
I_{11} &\triangleq \mathsf{C} (h_{11}^2 P_1 ),\qquad I_{12} \triangleq \mathsf{C} (h_{11}^2 P_1 + h_{21}^2 P_2 ), \\
I_{21} &\triangleq \mathsf{C} (h_{22}^2 P_2 ),\qquad   I_{22}  \triangleq \mathsf{C} (h_{22}^2 P_2 + h_{12}^2 P_1 ),  \\
\mathbf{I}_{\mathrm{c}} &\triangleq \left[I_{11}   \, \, I_{21}  \right]^T, \qquad  \mathbf{I}_{\mathrm{d}} \triangleq \left[ I_{11} \,\, I_{21} \,\, I_{12} \,\, I_{22} \right]^T. 
\end{align}
\end{mydef}
The vectors $\mathbf{I}_\mathrm{c}$ and $\mathbf{I}_{\mathrm{d}}$ characterize the first-order regions that are obtained naturally from  converse and direct bounds respectively. The non-asymptotic bounds that we evaluate also yield these first-order vectors.

Carleial \cite{Carleial75} proved  that the capacity region $\mathcal{C}$ of the Gaussian IC in the very strong interference regime is given by
\begin{align}
\mathcal{C} = \left\{(R_1,R_2) \in \mathbb{R}_+^2 \, | \  R_1 \leq I_{11},\,\, R_2 \leq I_{21} \right\}. \label{eqn:carl}
\end{align}

A certain set of information densities plays an important role for the  IC \cite{HK81,CMG,LTM13}. However, in dealing with channels with cost constraints, modified information densities \cite{Hayashi2009,ML13} offer certain advantages in the evaluation of non-asymptotic bounds as $n \to \infty$. 

\begin{mydef} \label{Di}
Fix a joint distribution
\begin{align}
P_{Y_1^n Y_2^n X_1^n X_2^n }(y_1^n y_2^n x_1^n x_2^n  ) = P_{X_1^n}(x_1^n)P_{X_2^n}(x_2^n) W_{1}^n (y_{1}^n|x_{1}^n x_{2}^n) W_{2}^n ( y_{2}^n|x_{1}^n x_{2}^n). \label{Ed}
\end{align}
Given two auxiliary   (conditional)  output distributions   $Q_{Y_1^n|X_2^n}$ and $Q_{Y_1^n}$  \footnote{In the following, we will refer to $Q_{Y_1^n|X_2^n}$ and $Q_{Y_1^n}$ collectively as output distributions, dropping the qualifier \textit{conditional}, for the sake of brevity.}, define the {\em modified information densities}
\begin{align}
\tilde{i}^n_{11}( X_1^n X_2^n Y_1^n) &\triangleq \log \frac{ W_{1}^n (Y_{1}^n|X_{1}^n X_{2}^n)}{Q_{Y_1^n|X_2^n}(Y_1^n|X_2^n)}, \label{eqn:i11_tilde} \\ 
\tilde{i}^n_{12}(X_1^n X_2^n Y_1^n) &\triangleq \log \frac{ W_{1}^n (Y_{1}^n|X_{1}^n X_{2}^n)}{Q_{Y_1^n}(Y_1^n)}. \label{eqn:i12_tilde}
\end{align}
We will often use the shorthands $\tilde{i}^n_{11}$ and $\tilde{i}^n_{12}$. Furthermore, the dependencies of $\tilde{i}^n_{11}$ and $\tilde{i}^n_{12}$ on the channel $W_1^n$ and the     output distributions   $Q_{Y_1^n|X_2^n}$ and  $Q_{Y_1^n}$ will be suppressed for the sake of brevity.
  
Similarly, given two auxiliary     output distributions   $Q_{Y_2^n|X_1^n}$ and $Q_{Y_2^n}$, we define $\tilde{i}^n_{21}( X_1^n X_2^n Y_2^n)$ and $\tilde{i}^n_{22}(X_1^n X_2^n Y_2^n)$. 
In addition, we define
\begin{align}
\tilde{\mathbf{i}} ^n_{\mathrm{c}}(X_1^n X_2^n Y_1^n Y_2^n) &\triangleq [\tilde{i}^n_{11}  \quad \tilde{i}^n_{21}]^T \label{eqn:ic_vec} \\
\tilde{\mathbf{i}} ^n_{\mathrm{d}} (X_1^n X_2^n Y_1^n Y_2^n) &\triangleq [\tilde{i}^n_{11} \quad \tilde{i}^n_{21} \quad \tilde{i}^n_{12} \quad \tilde{i}^n_{22} ]^T.  \label{eqn:id_vec}
\end{align}
\end{mydef}

\begin{mydef}
Recall the definition of the Gaussian dispersion function $\mathsf{V}(\cdot )$ in \eqref{eqn:disp_def}.  Define the second-order quantities
\begin{align}
V_1 \triangleq \mathsf{V}(h_{11}^2 P_1),\quad\mbox{and}\quad
V_2 \triangleq \mathsf{V}(h_{22}^2 P_2) .   
\end{align}
\end{mydef}
Note that $h_{jj}^2 P_j$ is the signal-to-noise ratio of the direct  channel from  $\mathrm{Tx}_j$ to  $\mathrm{Rx}_j$ and $\mathsf{V}(h_{jj}^2 P_j)$ is the corresponding dispersion. Also, the expectation and the conditional covariance of  the random vector $\tilde{\mathbf{i}}_{\mathrm{c}}(X_1  X_2  Y_1  Y_2)$ are $\mathbf{I}_{\mathrm{c}}$ and $\mathrm{diag}([V_1, V_2])$ respectively if $(X_1 , X_2 )\sim \mathcal{N}(\mathbf{0}, \mathrm{diag}( [ P_1, P_2]) )$, $Q_{Y_1|X_2}(\cdot|x_2) = \mathcal{N}(h_{21}x_2, h_{11}^2 P_1 +   1)$ and $Q_{Y_2|X_1} (\cdot|x_1) = \mathcal{N}(h_{12}x_1, h_{22}^2 P_2 +   1)$.

%

The following  is the cumulative distribution function of a standard Gaussian 
\begin{equation}
\Phi(t)\triangleq  \int_{-\infty}^t  \frac{1}{\sqrt{2\pi}}\exp(-u^2/2)\,\mathrm{d} u. 
\end{equation}
The inverse of $\Phi$ is defined as $\Phi^{-1}(\epsilon)\triangleq \sup\{t \in\mathbb{R} \, | \, \Phi(t)\le\epsilon\}$. 
 
In this paper, we   aim to characterize the $(\kappa_1,\kappa_2,\epsilon)$-capacity region  of the Gaussian IC in the strictly very strong interference regime, i.e., we determine $\mathcal{L}(\kappa_1,\kappa_2,\epsilon)$ for any $(\kappa_1,\kappa_2)\in [0,\infty)^2$ and $\epsilon\in (0,1)$.

\section{Main result} \label{MainResult}
The main result of this paper is summarized in  the following theorem. See Figure~\ref{fig:wz} for an illustration of the different cases.
\begin{mythm} \label{T1}
For any $0 < \epsilon < 1$, the $(\kappa_1, \kappa_2, \epsilon)$-second-order capacity region for the strictly very strong Gaussian interference channel in the following special cases is given by:

i) When $\kappa_1 = I_{11}$ and $\kappa_2 < I_{21}$ (vertical boundary),
\begin{align}
\mathcal{L}(\kappa_1, \kappa_2, \epsilon) = \left\{(L_1,L_2 ) \in\mathbb{R}^2 \bigg| \Phi\bigg(\frac{L_1}{\sqrt{V_1}}\bigg)\le\epsilon \right\};
\end{align}

ii) When $\kappa_1 = I_{11}$ and $\kappa_2 = I_{21}$ (corner point),
\begin{align} \label{eqn:case2}
\mathcal{L}(\kappa_1, \kappa_2, \epsilon)   =  \left\{  (L_1,L_2)\in\mathbb{R}^2\, \bigg|\,   \Phi\bigg(  - \frac{L_1}{\sqrt{V_1}}\bigg)  \Phi\bigg( - \frac{L_2}{\sqrt{V_2}}\bigg) \ge 1-\epsilon \right\};  
\end{align}

iii) When $\kappa_1 < I_{11}$ and $\kappa_2 = I_{21}$ (horizontal boundary),
\begin{align}
\mathcal{L}(\kappa_1, \kappa_2, \epsilon) = \left\{(L_1,L_2) \in\mathbb{R}^2\, \bigg| \,   \Phi\bigg(  \frac{L_2}{\sqrt{V_2}}\bigg)\le \epsilon \right\}.
\end{align}

\end{mythm}
\begin{proof}
This theorem is proved in the appendix.
\end{proof}


\begin{myexample}
We  visualize the result of case (ii) of Thereom \ref{T1} via an example. Consider a Gaussian IC where the dispersions are equal, i.e., $V_1=V_2$, and the average error probability $\epsilon = 0.001$. Clearly, by choosing $h_{12}$ and $h_{21}$ sufficiently large, we can guarantee that the Gaussian  IC is in the strictly  very strong interference regime (see Example~\ref{eg:first}). The second-order capacity region $\mathcal{L}(\kappa_1, \kappa_2, \epsilon)$ of case (ii)   where $(\kappa_1, \kappa_2) =( I_{11}, I_{21})$  is illustrated in Figure~\ref{fig:ex}. Because $\epsilon<1/2$,  the second-order capacity region $\mathcal{L}(\kappa_1, \kappa_2, \epsilon)$ lies entirely in the third quadrant of $\mathbb{R}^2$. Due to the fact that $V_1=V_2$, the second-order capacity region $\mathcal{L}(\kappa_1, \kappa_2, \epsilon)$ for case (ii) is also symmetric about the line $L_1 = L_2$. 
\end{myexample}

\begin{figure}
\centering
\includegraphics[width=0.57\textwidth]{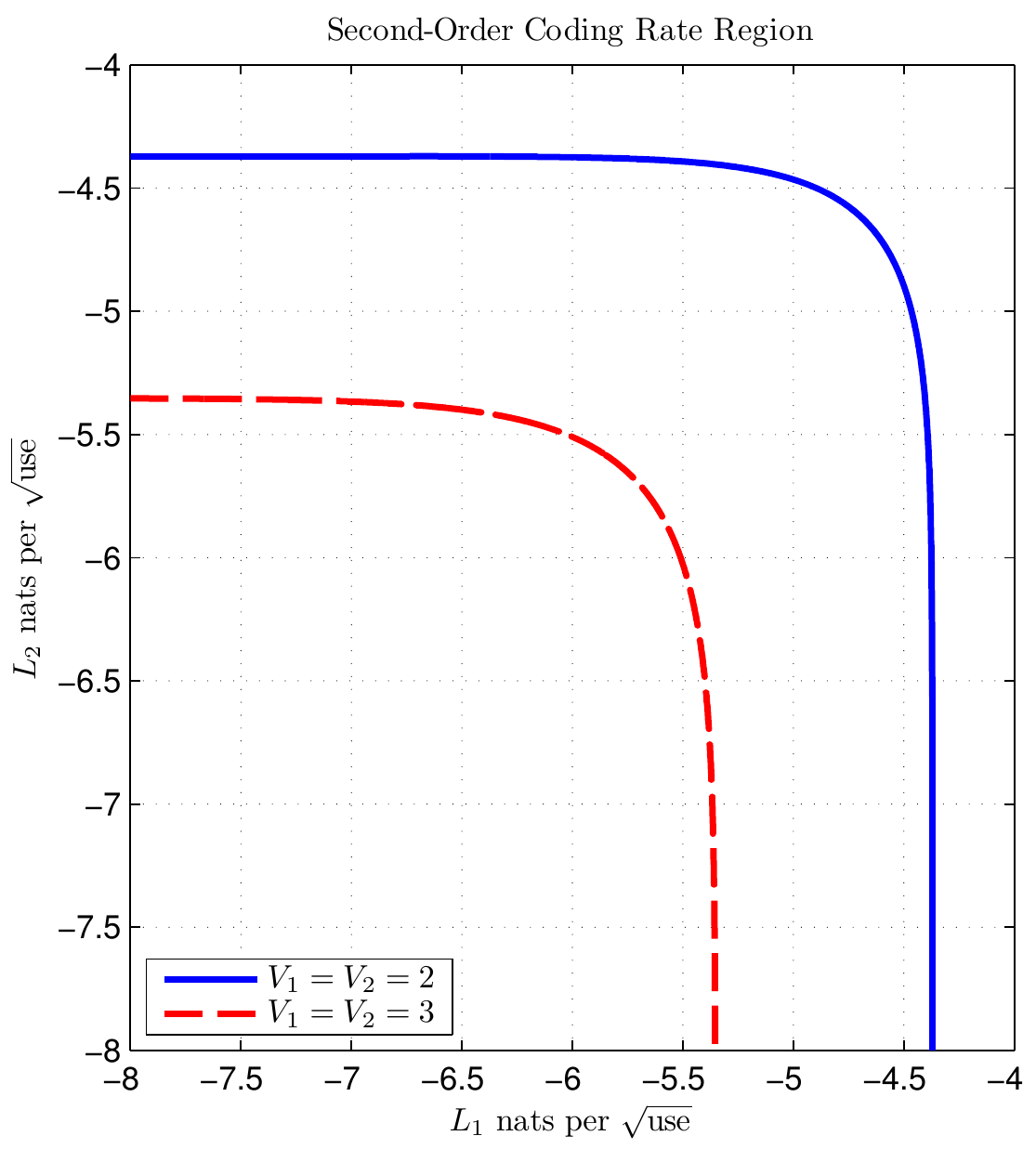}
\caption{The second-order capacity region $\mathcal{L}(\kappa_1, \kappa_2, \epsilon)$ of case $2$ when $\epsilon = 0.001$}.
\label{fig:ex}
\end{figure}

\subsection{Remarks Concerning Theorem~\ref{T1}}

\begin{enumerate}

\item  The result is applicable to any $(\kappa_1, \kappa_2) \in [0,\infty)^2$. If $(\kappa_1,\kappa_2)$ is  in the interior of $\mathcal{C}$, then it can be shown that $\mathcal{L}(\kappa_1, \kappa_2, \epsilon) = \mathbb{R}^2$. If $(\kappa_1,\kappa_2)$ is   in the exterior of   $\mathcal{C}$, then $\mathcal{L}(\kappa_1, \kappa_2, \epsilon) = \emptyset$. This implies the {\em strong converse}. Thus, the strong converse, which was hitherto not established for the Gaussian IC with very strong interference, is a by-product of our analyses.  The only interesting cases, in which $(\kappa_1, \kappa_2)$ is on the boundary of the capacity region, are presented in Theorem~\ref{T1}.

\item In case (i), the $(\kappa_1, \kappa_2,\epsilon)$-capacity region depends on $\epsilon$ and $V_1$ only. This region is more succinctly described as \begin{align}
L_1 \leq  \sqrt{V_1}\Phi^{-1}(\epsilon),\quad\mbox{and}\quad  L_2\in\mathbb{R}. \label{eqn:user1_disp} 
\end{align}
Note that $\sqrt{V_1}\Phi^{-1}(\epsilon)$ is exactly the second-order coding rate of the AWGN channel between transmitter $\mathrm{Tx}_1$ and receiver $\mathrm{Rx}_1$ when there is no interference from transmitter $\mathrm{Tx}_2$ \cite{Hayashi2009}. The fact that user $2$'s parameters do not feature in \eqref{eqn:user1_disp}   is because $\kappa_2 < I_{21}$. Note that $\kappa_2 < I_{21}$ implies that $\mathrm{Tx}_2$  operates  at a rate strictly below the   capacity of the second channel $I_{21}$. In this case, the second channel operates in the large-deviations (error exponents) regime so the second constraint is not featured in our dispersion analysis. This is because the error probability is exponentially small in this regime. See \cite{TanKosut12,Haim12, ST13, Nomura13}.  By symmetry, case (iii) is similar to case (i).

\item In case (ii),  the $(\kappa_1, \kappa_2,\epsilon)$-second-order capacity region is a function of $\epsilon$ and  {\em both} $V_1$ and $V_2$ because we are operating at rates near the {\em corner point} of $\mathcal{C}$. The two constraints on the rates come into play in the characterization of $\mathcal{L}(\kappa_1,\kappa_2,\epsilon)$. 
Roughly speaking, $\Phi(-L_j/\sqrt{V_j})$ is the probability that the $j^{\mathrm{th}}$-decoder decodes correctly if the number of codewords of the $j^{\mathrm{th}}$-user is 
\begin{equation}
M_{jn} = \big\lfloor\exp \big( n \kappa_j  + \sqrt{n}L_j + o(\sqrt{n } )  \big) \big\rfloor. \label{eqn:size_code}
\end{equation}
Thus, the product $\Phi(-L_1/\sqrt{V_1}) \Phi(-L_2/\sqrt{V_2})$, which is constrained to be larger than $1-\epsilon$ in \eqref{eqn:case2}, is the probability that {\em both} messages are decoded correctly assuming that both channels operate {\em independently}. More explicitly,  using the definition of the error probability criterion in~\eqref{eqn:error_prob}, we have that  
\begin{align}
 \Pr \left(\hat{S}_1 = S_1 \text{ and } \hat{S}_2 = S_2 \right)  \geq 1 - \epsilon .
 \end{align}
Assuming independence, this means that 
 \begin{align}
   \Pr \left(\hat{S}_1 = S_1 \right) \Pr \left(\hat{S}_2 = S_2 \right)  \geq 1 - \epsilon. 
   \end{align}
 Denoting $o(1)$ as a sequence that tends to zero as the blocklength grows, we observe that 
\begin{equation}
   \Pr\left(\hat{S}_j = S_j \right) =\Phi \bigg(-\frac{L_j}{\sqrt{V_j}}\bigg) + o(1)
   \end{equation}   
 if \eqref{eqn:size_code} holds  (a result by Hayashi \cite[Thm.~4]{Hayashi2009}). In this way, we recover  the main result in~\eqref{eqn:case2}.  Since $V_1=\mathsf{V}(h_{11}^2 P_1)$ and $V_2=\mathsf{V}(h_{22}^2 P_2)$ are the dispersions of the point-to-point Gaussian channels without interference, this is exactly analogous to Carleial's result for Gaussian ICs with very strong interference \cite{Carleial75}. In other words, in this regime, the channel dispersions of the constituent channels are not affected. This explains the title of the paper---namely that in this very special scenario, interference does not affect (reduce) the  dispersions of the constituent channels. In addition, no cross dispersion terms  are present in~\eqref{eqn:case2} unlike other network problems~\cite{TanKosut12, Nomura13, ST13}. This is due to the independence of the noises $Z_{1i}$ and $Z_{2i}$ as well as the strictly very strong interference assumption.


\item One of the  input distributions that achieves  the capacity, error exponent, dispersion and even the third-order coding rate of  the Gaussian point-to-point channel \cite{Shannon59, Polyanskiy2010, TanTom13}, is the uniform distribution on the  power sphere. MolavianJazi-Laneman~\cite{ML13} derived global achievable dispersions for the two-user  Gaussian MAC using uniform distributions on power spheres.  In this work,  we also use the {\em uniform input distributions on  power spheres}.  It is not easy to use the  {\em cost constrained ensemble} in~\cite{ST13} as that  input  distribution is more suited to, for example, superposition coding.

\item The proof of the  direct part makes use of a generalized version of Feinstein's lemma~\cite{Feinstein}, which involves four error events. We also use the {\em central limit theorem for functions} by MolavianJazi and Laneman~\cite{ML13} to ``lift'' the problem to a higher dimension, in fact $10$-dimensional Euclidean space,  ensuring that the i.i.d.\ version of the multivariate Berry-Esseen theorem \cite{Gotze91,BH10,WKT13} may be employed. The converse   makes use of a generalized version of the Verd\'u-Han lemma~\cite{VH94}, which involves only two error events.  At a high level, we use the strictly very strong interference condition to reduce the number of error events in the direct part, so that it matches the converse. For   ICs in the  very strong interference regime~\cite{Carleial75}, the intuition is that each receiver can reliably decode information from the non-intended transmitter. Interestingly, this intuition carries over for second-order (dispersion) analysis with the caveat that the interference must be {\em strictly} very strong. 

\item Finally, it is somewhat surprising that in the converse, even though we must ensure that the transmitter outputs are independent,  we do not  need to use the wringing technique, invented by Ahlswede~\cite{Ahl82}  and  used originally to prove that the DM-MAC admits a strong converse. This is due to Gaussianity which allows us to show that the first- and second-order  statistics of a certain set of information densities are independent of $x_1^n$ and $x_2^n$ on    power spheres.  See~\eqref{eqn:expect}-\eqref{eqn:cov1}.  
\end{enumerate}

\section{Reflections} \label{FutureWork}
In this work, we characterized the second-order coding rates of the Gaussian interference channel in the strictly very strong interference regime. The strictly very strong interference assumption reduces the number of error events in the direct part so that it matches the converse. It would be interesting to find the second-order capacity region  in the other regimes. New non-asymptotic achievability and converse  bounds are   needed  for other cases. In particular, it is intriguing to see what the second-order capacity region for the interference channel  in the strong interference regime is. Note that in the strong interference regime, the interference channel behaves like a pair of MACs but unfortunately the second-order capacity region for the MAC remains unknown \cite{TanKosut12,ML13,Huang12,ST13}. The achievability scheme in this work is also applicable to the interference channel in the strong interference regime. A non-trivial problem here is to derive a tighter  converse than that prescribed by Lemma \ref{LemmaVH} to be evaluated assuming only  strong interference.

The class of {\em mixed channels} forms an  important class of models  for theoretical study  as they are the canonical class of non-ergodic channels \cite{Han2003}. The second-order source coding rate region has been considered for the mixed correlated source for the Slepian-Wolf problem in \cite{Nomura13}. The corresponding point-to-point channel coding problem was also studied in \cite{PPV11, TomTan13b}. It would be also interesting to   find the second-order capacity region for the mixed Gaussian IC. The key difficulty is that characterizing the second-order capacity region for the mixed Gaussian IC appears to involve manipulating the modified information densities and the auxiliary     output distributions. Previous works in mixed channels in  \cite{Han2003, Nomura13} do not involve  auxiliary     output distributions. New achievability and converse techniques will be needed to find the second-order capacity region for the mixed Gaussian~IC.
 
Lastly, it appears that the corner point result in \eqref{eqn:case2} may be generalized to $3$ (or more) sender-receiver pairs simply by setting the product $\prod_{j} \Phi(-L_j/\sqrt{V_j})$ to be no smaller than $1-\epsilon$. The question then becomes: ``What is the appropriate generalization of the assumption of being in the strictly very strong interference regime to the  $3$ (or more) sender-receiver pair setting?''   


\section{Appendix}

\subsection{Proof of Theorem \ref{T1}: Converse Part}
In this subsection, we present the converse proof of Theorem \ref{T1}.  By a standard $n\leftrightarrow n+1$ argument \cite[Sec.~X]{Shannon59} \cite[Lem.~39]{Polyanskiy2010}, we may assume that the power constraints are satisfied with equality.  We first start with an non-asymptotic bound, which is a generalized version of Verd\'u-Han lemma~\cite[Lem.~4]{VH94}.   The proof of this lemma is  given in Subsection \ref{Sec:ProofVH}.
\begin{mylemma} \label{LemmaVH}
For every $n \in \mathbb{N}$, for every $\gamma >0$, and for any auxiliary     output distributions   $Q_{Y_1^n|X_2^n}$ and $Q_{Y_2^n|X_1^n}$, every $(M_{1n},M_{2n},n,\epsilon_n,P_1,P_2)$-code for the Gaussian IC satisfies
\begin{align}
\epsilon_n &\geq \Pr ( \tilde{i}_{11}^n (X_1^n X_2^n Y_1^n ) \leq \log M_{1n} - n\gamma  \nonumber \\
           &\quad \text{ or } \tilde{i}_{21}^n (X_1^n X_2^n Y_2^n) \leq \log M_{2n} - n\gamma ) - 2e^{-n\gamma},
\end{align}
where $\tilde{i}_{11}$ and $\tilde{i}_{21}$ are modified information densities defined in \eqref{eqn:i11_tilde} and \eqref{eqn:i12_tilde} respectively and   $X_j^n$ is uniformly distributed  over the $j^{\mathrm{th}}$  codebook and so $\|X_j^n\|^2 = nP_j$ with probability one. 
\end{mylemma}

\begin{myrm}
Intuitively, the proof of Lemma \ref{LemmaVH} relies on the fact that a system with help of a genie, which provides the transmitted information of   transmitter $2$ to decoder $1$, and the transmitted information from transmitter $1$ to decoder $2$, will always do no worse than a system without help from a genie.
\end{myrm}

Fix any pair of rates $(\kappa_1, \kappa_2) $ on the boundary of  $\mathcal{C}$ in \eqref{eqn:carl}. Consider any second-order pair $(L_1,L_2)$ that is $(\kappa_1, \kappa_2, \epsilon)$-achievable for the Gaussian IC. This implies that there exists a sequence of  $(M_{1n}, M_{2n},n,\epsilon_n,P_1,P_2)$-codes satisfying~\eqref{eqn:second_order}.

By the definition of $\liminf$, for any $\beta >0$, there exists an integer $N_\beta$ such that  for all $n > N_\beta$
\begin{align}
 \log M_{jn} - n \kappa_j \geq \sqrt{n}( L_j - \beta )   . \label{eqn:use_beta}
\end{align}


Let $\mathcal{L}_{\mathrm{eq}}(\kappa_1,\kappa_2,\epsilon)$ be the $(\kappa_1,\kappa_2,\epsilon)$-second-order capacity region of the IC with equal power constraints, i.e.\ each codeword $x_j^n$ satisfies $\sum _{k=1} ^{n}  x_{jk}^2 = nP_j$ for $j=1,2$. As mentioned above, it can  be shown that (cf.~\cite[Lem.~39]{Polyanskiy2010}) $\mathcal{L}_{\mathrm{eq}}(\kappa_1,\kappa_2,\epsilon) = \mathcal{L}(\kappa_1,\kappa_2,\epsilon)$. Therefore, in this converse proof, it is sufficient to assume   equal power constraints. \\

Define the auxiliary     output distributions   
\begin{align}
\hat{Q}_{Y_1|X_2}(y_1|x_2) &\triangleq \mathcal{N}(y_1; h_{21} x_2, h_{11}^2 P_1 + 1) \label{EQ1} \\
\hat{Q}_{Y_2|X_1}(y_2|x_1) &\triangleq \mathcal{N}(y_2; h_{12} x_1, h_{22}^2 P_2 + 1) .  \label{EQ2} 
\end{align}
These are the  conditional  output distributions   of the Gaussian IC when the inputs are $X_1\sim \mathcal{N}(0,P_1)$ and $X_2 \sim \mathcal{N}(0,P_2)$. 

Choose the  conditional  output distributions   $Q_{{Y}_1^n|{X}_2^n}$ and $Q_{{Y}_2^n|{X}_1^n}$  in Lemma \ref{LemmaVH}, respectively  as the $n$-fold products of $\hat{Q}_{Y_1|X_2}(y_1|x_2)$ and $\hat{Q}_{Y_2|X_1}(y_2|x_1)$, which are defined above. Next, choose $\gamma = \frac{\log n}{2n}$.  Let $V_{\mathrm{c}}$ be the $2\times 2$ diagonal matrix with $V_1$ and $V_2$ along its diagonals. 

Next, we have the following lemma whose proof is presented in full in Subsection \ref{SecCalC}. 
\begin{mylemma} \label{LemmaCalConverse}
For all $x_1^n$ and $x_2^n$ satisfying $\|x_j^n\|^2=nP_j$  we have
\begin{align}
\mathbb{E} \left[\frac{1}{n }\sum_{k=1}^n \tilde{\mathbf{i}}_{{\mathrm{c}}k} ( x_{1k} x_{2k} Y_{1k} Y_{2k})\right] &=   \mathbf{I}_{\mathrm{c}},\quad\mbox{and} \label{eqn:expect}\\
\mathrm{cov } \left[\frac{1}{\sqrt{n}}\sum_{k=1}^n  \tilde{\mathbf{i}}_{{\mathrm{c}}k} ( x_{1k} x_{2k} Y_{1k} Y_{2k})\right] &= V_{\mathrm{c}} ,\label{eqn:cov1}
\end{align}
where $\tilde{\mathbf{i}}_{{\mathrm{c}}k} $ is the random vector with components given by \eqref{eqn:i11_tilde} and \eqref{eqn:i12_tilde}.
\end{mylemma}
This lemma is the crux of the converse proof.  Note that the covariance matrix in \eqref{eqn:cov1} is diagonal and this results in the decoupling of the events in the corner point case given by~\eqref{eqn:case2}. The diagonal nature of \eqref{eqn:cov1} arises, in part, from the independence of the noises $Z_{1i}$ and $Z_{2i}$ for each time $i = 1,\ldots, n$. 

Let $t_{\mathrm{c}} \triangleq \frac{1}{n} \sum_{k=1}^n \mathbb{E} [\|\tilde{\mathbf{i}}_{{\mathrm{c}}k} ( x_{1k} x_{2k} Y_{1k} Y_{2k})\|^3]$ be the third absolute moment and $\phi_{\mathrm{c}} \triangleq \frac{254\sqrt{2}t_{\mathrm{c}}}{\lambda_{\min}(V_{\mathrm{c}})^{3/2}}$, where $\lambda_{\min}(V_{\mathrm{c}})$ is the minimum eigenvalue of $V_{\mathrm{c}}$.  
Define  the rate pair $\mathbf{R}_{\mathrm{c}} \triangleq [\frac{\log M_{1n}}{n} ,  \frac{\log M_{2n}}{n}]^T$.  Note that $V_{\mathrm{c}} \succ 0$ because the channel gains and powers are all positive. Also $t_{\mathrm{c}}<\infty$ from \cite[App.~A]{ST13}.  Thus, $\phi_{\mathrm{c}}$ is finite. Define
\begin{equation}
\Psi\big([t_1, t_2] ; \mathbf{m}, \bm{\Sigma} \big) \triangleq \int_{-\infty}^{t_1} \int_{-\infty}^{t_2} \mathcal{N}(\mathbf{u}; \mathbf{m}, \bm{\Sigma})\,\mathrm{d}\mathbf{u}
\end{equation}
as the bivariate generalization of the Gaussian cumulative distribution function.
Then we have
\begin{align}
 \Delta(x_1^n,x_2^n)
&\triangleq \Pr \left(\frac{1}{n} \sum_{k=1}^n \tilde{\mathbf{i}}_{{\mathrm{c}}k} ( x_{1k} x_{2k} Y_{1k} Y_{2k}) > \mathbf{R}_{\mathrm{c}} - \gamma \mathbf{1}  \right) \nonumber \\
&= \Pr \left(\frac{1}{\sqrt{n}}  \sum_{k=1}^n  \tilde{\mathbf{i}}_{{\mathrm{c}}k} - \sqrt{n}\mathbf{I}_{\mathrm{c}} > \sqrt{n}(\mathbf{R}_{\mathrm{c}} -\mathbf{I}_{\mathrm{c}} - \gamma \mathbf{1} ) \right)  \nonumber\\
&\stackrel{(a)}\leq \Psi(-\sqrt{n}(\mathbf{R}_{\mathrm{c}} -\mathbf{I}_{\mathrm{c}} - \gamma \mathbf{1}); \mathbf{0},V_{\mathrm{c}}) + \frac{\phi_{\mathrm{c}}}{\sqrt{n}} \nonumber\\
&\stackrel{(b)}\leq \Psi(-\sqrt{n}(\mathbf{R}_{\mathrm{c}} -\mathbf{I}_{\mathrm{c}} ); \mathbf{0}, V_{\mathrm{c}}) + O \left(\frac{\log n}{\sqrt{n}} \right), \label{E119}
\end{align}
where 
$(a)$ follows from the application of a variant of the multivariate Berry-Esseen Theorem, which is stated in Lemma \ref{LemmaBE}; 
and $(b)$ follows from Taylor expansion  of the function $\Psi(\mathbf{t}; \mathbf{0}, V_{\mathrm{c}})$, which is differentiable with respect to $\mathbf{t}$.

From Lemma \ref{LemmaVH}, we have 
\begin{align}
\epsilon_n &\geq 1 - \Pr \left(\frac{1}{n}\tilde{\mathbf{i}}^n_{\mathrm{c}} (X_{1}^n X_{2}^n Y_{1}^n Y_{2}^n ) > \mathbf{R}_{\mathrm{c}}- \gamma \mathbf{1} \right) - 2 e^{-n\gamma} \nonumber\\
           &= 1 - \mathbb{E} \left[ \Delta(X_1^n,X_2^n) \right] - 2 e^{-n\gamma}. \label{E121}
\end{align}
Note that $e^{-n\gamma} = \frac{1}{\sqrt{n}}$. Combining (\ref{E119}) and (\ref{E121}), we have
\begin{align}
\epsilon_n &\geq 1 - \Psi(-\sqrt{n}(\mathbf{R}_{\mathrm{c}}-\mathbf{I}_{\mathrm{c}} ); \mathbf{0}, V_{\mathrm{c}}) - O\left(\frac{\log n}{\sqrt{n}}\right) - \frac{2}{\sqrt{n}} \nonumber\\
           &\stackrel{(a)}\geq 1 - \Psi\left(\begin{bmatrix}
                           \sqrt{n}(I_{11} - \kappa_1) - L_1 + \beta \\
                           \sqrt{n}(I_{21} - \kappa_2) - L_2 + \beta \\   
                          \end{bmatrix}; \mathbf{0}, V_{\mathrm{c}} \right) 
                           - O\left(\frac{\log n}{\sqrt{n}}\right) - \frac{2}{\sqrt{n}}  \label{Econ2}
\end{align}
where (a) holds for all $n>N_\beta$ and follows because  $\mathbf{t}\mapsto\Psi(\mathbf{t}; \mathbf{0}, V_{\mathrm{c}})$ is monotonically increasing in $\mathbf{t}$ and (\ref{eqn:use_beta}). We now consider  three different cases.\\

\textbf{Case} $1$: When $\kappa_1 = I_{11}$ and $\kappa_2 < I_{21}$ \\
For any fixed $L_2$, if $\kappa_2 < I_{21}$, we have $\sqrt{n}(I_{21} - \kappa_2) - L_2 + \beta \to +\infty$. Thus, the second term on the RHS of~(\ref{Econ2}) converges to $\Psi \left(- L_1 + \beta ; 0, V_1 \right)  = \Phi\big( \frac{-L_1+\beta}{\sqrt{V_1}}\big)$. Taking $\limsup$ on both sides of~(\ref{Econ2}), and using~(\ref{eqn:use_beta}), we have
\begin{align}
\epsilon \geq \limsup_{n \to \infty} \epsilon_n \geq 1- \Phi\bigg(\frac{ -L_1+\beta}{ \sqrt{V_1}}\bigg).
\end{align}
Since this is true for any $\beta >0$,  we may let $\beta\downarrow 0$ and deduce that  
\begin{equation}
\Phi\bigg( \frac{L_1}{\sqrt{V_1}}\bigg)\le\epsilon.
\end{equation}
Case $1$ is proved.\\

\textbf{Case} $2$: When $\kappa_1 = I_{11}$ and $\kappa_2 = I_{21}$\\
In this case,  the second term on the RHS of (\ref{Econ2}) converges to $\Psi \left([- L_1 + \beta,-L_2 +\beta]^T ; 0, V_{\mathrm{c}} \right)$. The rest of the arguments are similar to that in case $1$. Note that because $V_{\mathrm{c}}$ is diagonal, 
\begin{equation}
\Psi \left([- L_1,-L_2]^T ; 0, V_{\mathrm{c}} \right) = \Phi\bigg(-\frac{L_1}{\sqrt{V_1}}\bigg)\Phi\bigg(-\frac{L_2}{\sqrt{V_2}}\bigg). \label{eqn:diag}
\end{equation}

\textbf{Case} $3$: When $\kappa_1 < I_{11}$ and $\kappa_2 = I_{21}$  \\
By symmetry, case $3$ is proved similarly to case $1$. \\

%

\subsection{Proof of Theorem \ref{T1}: Direct Part}
In this subsection, we present the  achievability proof of Theorem \ref{T1}. The following non-asymptotic bound,  a generalized version of Feinstein's lemma~\cite{Feinstein}, will be employed in the proof. The proof of this lemma is given in Subsection \ref{Sec:ProofF}.

\begin{mylemma}  \label{LemmaF}
Fix a joint distribution satisfying (\ref{Ed}). For any $n \in\mathbb{N}$, any $\gamma>0$, and any auxiliary     output distributions   $Q_{Y_1^n|X_2^n}$, $Q_{Y_1^n}$, $Q_{Y_2^n|X_1^n}$ and $Q_{Y_2^n}$, there exists an $(M_{1n},M_{2n},n, \epsilon_n,P_1,P_2)$-code for the Gaussian IC,  such that
\begin{align}
\epsilon_n &\leq \Pr (\mathcal{E}_{11} \cup \mathcal{E}_{12} \cup \mathcal{E}_{21} \cup \mathcal{E}_{22}) + K e ^{-n\gamma} 
         + P_{X_1^n} ( \mathcal{F}_{1n}^c) +  P_{X_2^n} (\mathcal{F}_{2n}^c)
\end{align}
where
\begin{align}
\mathcal{E}_{11} \triangleq \{ \tilde{i}^n_{11}(X_1^n X_2^n Y_1^n) &\leq \log M_{1n} + n\gamma \} \\
\mathcal{E}_{21} \triangleq \{ \tilde{i}^n_{21}(X_1^n X_2^n Y_2^n) &\leq \log M_{2n} + n\gamma \} \\
\mathcal{E}_{12} \triangleq \{ \tilde{i}^n_{12}(X_1^n X_2^n Y_1^n) &\leq \log M_{1n} M_{2n} + n\gamma \} \\
\mathcal{E}_{22} \triangleq \{ \tilde{i}^n_{22}(X_1^n X_2^n Y_2^n) &\leq \log M_{1n} M_{2n} + n\gamma \}, 
\end{align}
and 
\begin{align}
K      &\triangleq K_{11} + K_{12} + K_{21} + K_{22},  \label{eqn:defK}\\
K_{11} &\triangleq \sup_{x_2^n , y_1^n } \frac{P_{Y_1^n|X_2^n}(y_1^n|x_2^n)}{Q_{Y_1^n|X_2^n}(y_1^n|x_2^n)},\qquad 
K_{12} \triangleq \sup_{y_1^n } \frac{P_{Y_1^n}(y_1^n)}{Q_{Y_1^n}(y_1^n)}, \\
K_{21} &\triangleq \sup_{x_1^n , y_2^n } \frac{P_{Y_2^n|X_1^n}(y_2^n|x_1^n)}{Q_{Y_2^n|X_1^n}(y_2^n|x_1^n)},\qquad 
K_{22} \triangleq \sup_{y_2^n } \frac{P_{Y_2^n}(y_2^n)}{Q_{Y_2^n}(y_2^n)}.  \label{eqn:defK3}
\end{align}
\end{mylemma}
\begin{myrm}
In fact, this lemma holds not just for Gaussian ICs, but for  general  ICs.
\end{myrm}
\begin{myrm}
The presence of the Radon-Nikodym derivatives $K_{ij}$ in \eqref{eqn:defK}--\eqref{eqn:defK3} is the price to pay for the luxury of using the auxiliary     output distributions. This version of generalized Feinstein is different from the earlier versions (cf.~\cite[Thm.~1]{VH94}) in that the information densities in this lemma involve auxiliary     output distributions   that can be chosen. This technique was similarly employed in \cite{Hayashi2009,ML13, Hayashi03}. By choosing  the appropriate auxiliary     output distributions   and input distributions, we can show that the inner bound to $\mathcal{L}(\kappa_1,\kappa_2,\epsilon)$ coincides with the outer bound. 
 \end{myrm}

First, we present the achievability proof for case $1$. \\
\textbf{Case} $1$: When $\kappa_1 = I_{11}$ and $\kappa_2 < I_{21}$ \\
Fix any  pair $(L_1, L_2)$ satisfying  
\begin{align}
\Phi\bigg( \frac{L_1}{\sqrt{V_1}}\bigg)\le\epsilon. \label{eqn:fixed_L1L2}
\end{align}
Let the number of codewords in the $j^{\mathrm{th}}$ codebook be 
\begin{align}
M_{nj}  = \lfloor \exp\big( n\kappa_j + \sqrt{n} L_j +  n^{1/4} \beta \big)  \rfloor
\end{align}
for $j=1,2,$ and a fixed $\beta >0$. It is clear that 
\begin{align}
\liminf_{n \to \infty} \frac {1}{\sqrt{n}} (\log M_{jn} - n \kappa_j) \geq L_j.
\end{align}

Therefore, in order to show that $(L_1,L_2)$ is $(\kappa_1,\kappa_2,\epsilon)$-achievable,  it suffices to show the existence of a sequence of $(M_{1n}, M_{2n},n,\epsilon_n,P_1,P_2)$-codes such that $\limsup_{n \to \infty}\epsilon_n \leq \epsilon$. For this, we define an appropriate input distribution to be used in Lemma~\ref{LemmaF}, which is going to be applied in this subsection.
Inspired by \cite{ML13, TanTom13}, we define the input distributions to be uniform on the respective power shells, i.e.\
\begin{align}
P_{X_j^n }(x_j^n  )\triangleq \frac{\delta( \|x_j^n\|-\sqrt{nP_j})}{A_n(\sqrt{nP_j})} , \label{EDinput}
\end{align}
for $j =1,2$ and  where $\delta(\cdot)$ is the Dirac delta and $A_n(r) \triangleq \frac{2 \pi^{n/2}}{\Gamma(n/2)}r^{n-1}$ is the surface area of a sphere in $\mathbb{R}^n$ with radius $r$.  With this choice, we have $P_{X_1^n} ( \mathcal{F}_{1n}^c)+ P_{X_2^n} (\mathcal{F}_{2n}^c)=0$, i.e.\ the power constraints are satisfied with probability $1$. 

Define the     output distributions   
\begin{align}
\hat{Q}_{Y_1}(y_1) &\triangleq \mathcal{N}(y_1; 0, h_{11}^2 P_1 + h_{12}^2 P_2 +1 ) \label{EQd1} \\
\hat{Q}_{Y_2}(y_2) &\triangleq \mathcal{N}(y_2; 0, h_{12}^2 P_1 + h_{22}^2 P_2 +1 ) \\
\hat{Q}_{Y_1|X_2}(y_1|x_2) &\triangleq \mathcal{N}(y_1; h_{21} x_2, h_{11}^2 P_1 + 1) \\
\hat{Q}_{Y_2|X_1}(y_2|x_1) &\triangleq \mathcal{N}(y_2; h_{12} x_1, h_{22}^2 P_2 + 1) .  \label{EQd4} 
\end{align}
These are the     output distributions   of the Gaussian IC when the inputs are $X_1\sim \mathcal{N}(0,P_1)$ and $X_2 \sim \mathcal{N}(0,P_2)$. 

Choose the auxiliary     output distributions   $Q_{Y_1^n}(y_1^n) $, $Q_{Y_2^n}(y_2^n) $, $Q_{Y_1^n|X_2^n}(y_1^n|x_2^n) $ and $Q_{Y_2^n|X_1^n}(y_2^n|x_1^n) $ in Lemma~\ref{LemmaF} to be the $n$-fold memoryless extensions of $\hat{Q}_{Y_1}(y_1) $, $\hat{Q}_{Y_2}(y_2) $, $\hat{Q}_{Y_1|X_2}(y_1|x_2) $ and $\hat{Q}_{Y_2|X_1}(y_2|x_1) $ respectively, the distributions of which are given in (\ref{EQd1}-\ref{EQd4}). With this choice of auxiliary     output distributions, the value of $K$ in Lemma \ref{LemmaF} is shown in the following lemma to be bounded.

\begin{mylemma} \label{LemmaFiniteK}
For $n$ sufficiently large,  $K_{11}$, $K_{21}$, $K_{12}$ and $K_{22}$  are finite  . Thus, $K$ in (\ref{eqn:defK})  is also finite.
\end{mylemma}
This lemma is proved  in Subsection \ref{SecFiniteK}.

Define
\begin{align}
\alpha_{11} &\triangleq 1+ h_{11}^2 P_1, \quad \alpha_{12} \triangleq 1+ h_{11}^2 P_1 + h_{21}^2 P_2, \\
\alpha_{21} &\triangleq 1+ h_{22}^2 P_2, \quad \alpha_{22} \triangleq 1+ h_{12}^2 P_1 + h_{22}^2 P_2.
\end{align}
We have
\begin{align}
\tilde{i}_{11}^n &= \log \frac{W_1^n(Y_1^n|X_1^n X_2^n)}{Q_{Y_1^n|X_2^n}(Y_1^n|X_2^n)} \\
                 &= \frac{n}{2} \log (1 + h_{11}^2 P_1) + \frac{ \sum_{k=1}^n (Y_{1k} - h_{21} X_{2k})^2} {2(1+h_{11}^2P_1)} - \frac{\sum_{k=1}^n(Y_{1k}- h_{11} X_{1k} - h_{21} X_{2k})^2 } {2}  \\
                 &= \frac{n}{2} \log (1 + h_{11}^2 P_1) + \frac{ \sum_{k=1}^n (Z_{1k} + h_{11} X_{1k})^2} {2(1+h_{11}^2P_1)} - \frac{\sum_{k=1}^n(Z_{1k})^2} {2}  \\
                 &=n I_{11} + \frac{1}{2 \alpha_{11}}[(\alpha_{11}-1)(n- \|Z_1^n\|^2) + 2 h_{11}\langle X_1^n,Z_1^n \rangle], 
\end{align}
where $\langle a^n, b^n \rangle $ denotes the inner product between $a^n$ and $b^n$.

Similarly, it can be shown that the  other three modified information densities can be expressed as
\begin{align}
\tilde{i}_{21}^n &= n I_{21} + \frac{1}{2 \alpha_{21}}[(\alpha_{21}-1)(n- \|Z_2^n\|^2) + 2 h_{22}\langle X_2^n,Z_2^n \rangle ] \nonumber \\
\tilde{i}_{12}^n &= n I_{12} + \frac{1}{2 \alpha_{12}}[(\alpha_{12}-1)(n- \|Z_1^n\|^2)  \nonumber \\
         &\,\,+ 2 h_{11} h_{21} \langle X_2^n,X_1^n \rangle + 2 h_{11}\langle X_1^n,Z_1^n \rangle  + 2 h_{21}\langle X_2^n,Z_1^n \rangle ]  \nonumber \\
\tilde{i}_{22}^n &= n I_{22} + \frac{1}{2 \alpha_{22}}[(\alpha_{22}-1)(n- \|Z_2^n\|^2)  \nonumber \\
         &\,\,+ 2 h_{22} h_{12} \langle X_2^n,X_1^n \rangle + 2 h_{22}\langle X_2^n,Z_2^n \rangle  + 2 h_{12}\langle X_1^n,Z_2^n \rangle ].  
\end{align}

Next, we use the  {\em central limit theorem for functions} technique proposed by  MolavianJazi-Laneman~\cite{ML13} to transform these modified information densities into functions of sums of independent random vectors. Let $T_j^n \sim \mathcal{N}(\mathbf{0}, \mathbf{I}_{n\times n})$, for $j=1,2,$ be standard Gaussian random vectors that are independent of each other   and of  the noises $Z_j^n$. Note that the input distribution in (\ref{EDinput}) results in $X_{jk} = \sqrt{nP_j}\frac{T_{jk}}{\|T_j^n \|}$, for $k \in \{1,\ldots,n\}$.  Indeed, $\|X_j^n\|^2 = nP_j$ with probability one. Now consider the  length-$10$  random vector $\mathbf{U}_k \triangleq ( \{U_{j1k} \}_{j=1}^4,\{U_{j2k} \}_{j=1}^4,  U_{9k}, U_{10k})$, where  
\begin{alignat}{2}
&U_{11k} \triangleq 1 - Z_{1k}^2,\quad   & &  U_{21k}  \triangleq h_{11}\sqrt{P_1}T_{1k}Z_{1k},\, \nonumber \\
&U_{31k} \triangleq h_{21}\sqrt{P_2}T_{2k}Z_{1k}, \quad  && U_{41k} \triangleq h_{11}h_{21} \sqrt{P_1 P_2}T_{1k} T_{2k}, \, \nonumber \\
&U_{12k} \triangleq 1 - Z_{2k}^2,\quad  &&U_{22k} \triangleq h_{22}\sqrt{P_2}T_{2k}Z_{2k}, \, \nonumber \\
&U_{32k} \triangleq h_{12}\sqrt{P_1} T_{1k}Z_{2k},\quad &&U_{42k} \triangleq h_{12}h_{22} \sqrt{P_1 P_2}T_{1k} T_{2k}, \, \nonumber \\
&U_{9k} \text{ } \triangleq T_{1k}^2 -1,\quad && U_{10k} \triangleq T_{2k}^2 -1.
\end{alignat}
It is easy to verify that  $\mathbf{U}_k$ is i.i.d.\ across all channel uses $k \in \{1,\ldots,n\}$, and $\mathbb{E}(\mathbf{U}_k) =0$ and $\mathbb{E} (\|\mathbf{U}_k\|^3)$ is finite. The covariance matrix of $\mathbf{U}_1$ is given by
\begin{align}
\mathrm{Cov}(\mathbf{U}_1 ) = \begin{bmatrix}
2        &0                &0                &0               &0        &0              &0            &0               &0  &0\\
0        &\alpha_{11} -1   &0                &0               &0        &0              &0            &0               &0  &0\\
0        &0                &\alpha_{33}      &0               &0        &0              &0            &0               &0  &0\\
0        &0                &0                &\alpha_{44}     &0        &0              &0            &\alpha_{48}     &0  &0\\
0        &0                &0                &0               &2        &0              &0            &0               &0  &0\\
0        &0                &0                &0               &0        &\alpha_{21}-1  &0            &0               &0  &0\\
0        &0                &0                &0               &0        &0              &\alpha_{77}  &0               &0  &0\\
0        &0                &0                &\alpha_{48}     &0        &0              &0            &\alpha_{88}     &0  &0\\
0        &0                &0                &0               &0        &0              &0            &0               &2  &0\\
0        &0                &0                &0               &0        &0              &0            &0               &0  &2\\
\end{bmatrix},
\end{align} 
where
\begin{align}
\alpha_{33} &\triangleq h_{21}^2 P_2 \\
\alpha_{44} &\triangleq h_{11}^2 h_{21}^2 P_1 P_2 \\
\alpha_{48} &\triangleq P_1P_2h_{11}h_{21}h_{12}h_{22} \\
\alpha_{77} &\triangleq h_{12}^2 P_1\\ 
\alpha_{88} &\triangleq h_{12}^2 h_{22}^2 P_1 P_2.
\end{align}
Note that $\alpha_{11} + \alpha_{33} = \alpha_{12}$ and $\alpha_{21} + \alpha_{77} = \alpha_{22}$.\\

Define the functions $\tau_{11},\tau_{12}:\mathbb{R}^{10}\to \mathbb{R}$ as follows
\begin{align}
\tau_{11}(\mathbf{u}) &\triangleq (\alpha_{11} -1) u_{11} + \frac{2 u_{21}}{\sqrt{1 + u_9}} \\
\tau_{12}(\mathbf{u}) &\triangleq (\alpha_{12} -1) u_{11} + \frac{2 u_{21}}{\sqrt{1 + u_9}} + \frac{2 u_{31}}{\sqrt{1 + u_{10}}}  
                                     + \frac{2 u_{41}}{\sqrt{1 + u_9} \sqrt{1 + u_{10}}},
\end{align}
for receiver $1$. Similarly, define $\tau_{21}(\mathbf{u})$ and $\tau_{22}(\mathbf{u})$ for receiver~$2$ as follows
\begin{align}
\tau_{21}(\mathbf{u}) &\triangleq (\alpha_{21} -1) u_{12} + \frac{2 u_{22}}{\sqrt{1 + u_{10}}} \\
\tau_{22}(\mathbf{u}) &\triangleq (\alpha_{22} -1) u_{12} + \frac{2 u_{22}}{\sqrt{1 + u_{10}}} + \frac{2 u_{32}}{\sqrt{1 + u_{9}}}  
                                     + \frac{2 u_{42}}{\sqrt{1 + u_9} \sqrt{1 + u_{10}}}.
\end{align}
Denote $\mathbf{\tau}(\mathbf{u}) \triangleq [\tau_{11}(\mathbf{u}),\tau_{21}(\mathbf{u}) ,\tau_{12}(\mathbf{u}),\tau_{22}(\mathbf{u})]^T$. 
It can be shown that, for $l \in \{11,12,21,22\}$,
\begin{align}
\tilde{i}_l^n = n  I_l + \frac{n}{2 \alpha_l} \tau_l \left(\frac{1}{n}  \sum_{k=1}^n \mathbf{U}_k\right).
\end{align}

Denote the diagonal matrix $\Lambda \triangleq \mathrm{diag} (\frac{1}{\alpha_{11}}, \frac{1}{\alpha_{21}}, \frac{1}{\alpha_{12}}, \frac{1}{\alpha_{22}})$. We have
\begin{align}
\frac{1}{\sqrt{n}}\tilde{\mathbf{i}}_{\mathrm{d}}^n - \sqrt{n}\,  \mathbf{I}_{\mathrm{d}} =  \frac{\sqrt{n}}{2} \Lambda \mathbf{\tau} \left(\frac{1}{n}  \sum_{k=1}^n \mathbf{U}_k\right).
\end{align}

Note that $\mathbf{\tau}(\mathbf{0}) = \mathbf{0}$ and the vector function $\mathbf{\tau}(\mathbf{u})$ has continuous second-order derivatives in all neighbourhood of $\mathbf{u} =\mathbf{0}$. Therefore, the vector function $\mathbf{\tau}(\mathbf{u})$ satisfies the conditions given in Lemma \ref{LemmaBE2}. The Jacobian matrix $J_{\tau}(\mathbf{u})$ of $\mathbf{\tau}(\mathbf{u})$ with respect to $\mathbf{u}$, calculated at $\mathbf{u} = \mathbf{0}$, is given by 
\begin{align}
J_{\tau}(\mathbf{0}) = \begin{bmatrix}
\alpha_{11} -1   &2   &0  &0     &0                 &0  &0 &0    &0 &0\\
0                &0   &0  &0     &\alpha_{21} -1    &2  &0 &0    &0 &0\\
\alpha_{12} -1   &2   &2  &2     &0                 &0  &0 &0    &0 &0\\
0                &0   &0  &0     &\alpha_{22} -1    &2  &2 &2    &0 &0\\
\end{bmatrix}.
\end{align}
Next, by Lemma \ref{LemmaBE2}, we have that  the random vector $\frac{1}{\sqrt{n}}\tilde{\mathbf{i}}_{\mathrm{d}}^n - \sqrt{n}\,  \mathbf{I}_{\mathrm{d}}$ converges in distribution to a zero-mean Gaussian with  covariance matrix $V_{\mathrm{d}}$, which is given by
\begin{align}
V_{\mathrm{d}} &= \frac{1}{n} \cdot \frac{n}{4}\cdot  \Lambda J_{\tau}(\mathbf{0}) \mathrm{Cov}(\mathbf{U}_1 ) [J_{\tau}(\mathbf{0})]^T \Lambda \\
               &= \begin{bmatrix}
V_1                &0                   &V_{{\mathrm{d}}13}  &0   \\
0                  &V_2                 &0                   &V_{{\mathrm{d}}24} \\
V_{{\mathrm{d}}13} &0                   &V_{{\mathrm{d}}33}  &V_{{\mathrm{d}}34} \\
0                  &V_{{\mathrm{d}}24}  &V_{{\mathrm{d}}34}  &V_{{\mathrm{d}}44}  \\
\end{bmatrix} 
\end{align}
where 
\begin{align}
V_{{\mathrm{d}}13} &\triangleq \mathsf{V}(h_{11}^2 P_1, h_{11}^2 P_1 + h_{21}^2 P_2) \\
V_{{\mathrm{d}}24} &\triangleq \mathsf{V}(h_{22}^2 P_2, h_{22}^2 P_2 + h_{12}^2 P_1)\\
V_{{\mathrm{d}}33} &\triangleq \mathsf{V}(h_{11}^2P_1 + h_{21}^2 P_2)  + \frac{h_{11}^2 P_1 h_{21}^2 P_2}{(h_{11}^2 P_1 + h_{21}^2 P_2+1)^2}  \\
V_{{\mathrm{d}}44} &\triangleq \mathsf{V}(h_{22}^2 P_2 + h_{12}^2 P_1) + \frac{h_{12}^2 P_1 h_{22}^2 P_2}{(h_{12}^2 P_1 + h_{22}^2 P_2+1)^2}\\
V_{{\mathrm{d}}34} &\triangleq \frac{h_{12} h_{11} P_1  h_{21}h_{22} P_2}{(h_{11}^2 P_1 + h_{21}^2 P_2+1)(h_{12}^2 P_1 + h_{22}^2 P_2+1)}.
\end{align}

Thus, $V_{\mathrm{d}}$ has the form
\begin{align}
V_{\mathrm{d}} = \begin{bmatrix}
V_1 &0 &* &* \\
0 &V_2&* &* \\
* &* &*  &* \\
* &* &*  &*
\end{bmatrix}. \label{eqn:Vd}
\end{align}
In the above,  the $*$'s represent  entries that are inconsequential  for the purposes of subsequent analyses. 

Define the length-$4$ rate vector $\mathbf{R}_{\mathrm{d}} \triangleq  [\frac{\log M_{1n}}{n}  ,  \frac{\log M_{2n}}{n}, \frac{\log (M_{1n}M_{2n})}{n} , \frac{\log (M_{1n}M_{2n} ) }{n} ]^T$.
Appealing~to Lemma~\ref{LemmaF}, with $\gamma = \frac{\log n}{2n}$,  we have 
\begin{align}
\epsilon_n &\leq 1    -  \Pr \left(\frac{1}{\sqrt{n}}\tilde{\mathbf{i}}_{\mathrm{d}}^n ( X_{1}^n X_{2}^n Y_{1}^n Y_{2}^n)   >  \sqrt{n}( \mathbf{R}_{\mathrm{d}}   +   \gamma \mathbf{1} )\right)   -   K e^{-n\gamma} \nonumber \\
           &\leq 1 -  \Pr \left(\frac{1}{\sqrt{n}}  \sum_{k=1}^n  (\tilde{\mathbf{i}}_{{\mathrm{d}}k}   - \mathbf{I}_{\mathrm{d}} ) 
  >  \sqrt{n}\left(\mathbf{R}_{\mathrm{d}}   - \mathbf{I}_{\mathrm{d}}  +  \gamma \mathbf{1}  \right) \right) - \frac{K}{\sqrt{n}}\nonumber  \\
&\stackrel{(a)}\leq 1 -  \Psi \left(-\sqrt{n} \left(\mathbf{R}_{\mathrm{d}} -\mathbf{I}_{\mathrm{d}} + \gamma \mathbf{1} \right); \mathbf{0}, V_{\mathrm{d}} \right)  - O\left(\frac{1}{\sqrt{n}}\right)\nonumber\\
&\stackrel{(b)}\leq 1 - \Psi(-\sqrt{n}(\mathbf{R}_{\mathrm{d}} -\mathbf{I}_{\mathrm{d}} ); \mathbf{0}, V_{\mathrm{d}}) + O \left( \frac{\log n} {\sqrt{n}} \right)  ,\label{Edirect1}
\end{align}
where $(a)$ follows from a variant of the multivariate Berry-Esseen theorem, which is stated in Lemma \ref{LemmaBE2}; 
and $(b)$ follows from Taylor expanding $\mathbf{t}\mapsto\Psi(\mathbf{t}; \mathbf{0}, V_{\mathrm{d}})$.

Due to the strictly very strong interference  assumption (Definition~\ref{Def2}), 
\begin{equation}
h_{22}^2 P_2 +1 < \frac{h_{21}^2P_2 +h_{11}^2 P_1 + 1}{h_{11}^2 P_1 + 1}.
\end{equation}
Thus, $I_{11} + I_{21} <I_{12}$. Similarly, we have $I_{11} + I_{21} <I_{22}$. Therefore, as $n\to\infty$, we have 
\begin{align}
-\sqrt{n}(\mathbf{R}_{\mathrm{d}} -\mathbf{I}_{\mathrm{d}} ) 
= 
-\sqrt{n} \begin{bmatrix}
\kappa_1 +\frac{L_1}{\sqrt{n}} + \frac{\beta}{n^{3/4}} - I_{11} \\
\kappa_2 + \frac{L_2}{\sqrt{n}} + \frac{\beta}{n^{3/4}} - I_{21} \\
\kappa_1 +\kappa_2 + \frac{L_1}{\sqrt{n}} + \frac{L_2}{\sqrt{n}} + 2 \frac{\beta}{n^{3/4}} - I_{12} \\
\kappa_1 + \kappa_2 + \frac{L_1}{\sqrt{n}} + \frac{L_2}{\sqrt{n}} + 2 \frac{\beta}{n^{3/4}} - I_{22} \\
\end{bmatrix}
 \to   \begin{bmatrix}
-L_1 \\
+\infty \\
+\infty \\
+\infty \\
\end{bmatrix} \label{eqn:vectors} . 
\end{align}
Thus, $\Psi(-\sqrt{n}(\mathbf{R}_{\mathrm{d}}  - \mathbf{I}_{\mathrm{d}} ); \mathbf{0}, V_{\mathrm{d}})  \to  \Psi(-L_1;0,V_1)  =  \Phi\big(-\frac{L_1}{\sqrt{V_1}}\big)$. Taking $\limsup$ on both sides of (\ref{Edirect1}), we have 
\begin{align}
\limsup_{n \to \infty} \epsilon_n &\leq 1 - \Phi\bigg(-\frac{L_1}{\sqrt{V_1}}\bigg)  = \Phi\bigg( \frac{L_1}{\sqrt{V_1}}\bigg)  
                            \le  \epsilon, 
\end{align}
where the final inequality follows the choice of $L_1$ in~\eqref{eqn:fixed_L1L2}.  This completes the proof of the direct part for  Case 1.\\

\textbf{Case} $2$: When $\kappa_1 = I_{11}$ and $\kappa_2 = I_{21}$ \\ 
In this case, we have $\Psi(-\sqrt{n}(\mathbf{R}_{\mathrm{d}} -\mathbf{I}_{\mathrm{d}} ); \mathbf{0}, V_{\mathrm{d}}) \to \Psi([-L_1 \,   -L_2]^T;0,V_{\mathrm{c}})$ because the second and third entries   in \eqref{eqn:vectors}   tend  to $+\infty$  (by the strictly very strong interference assumption) while the first and fourth entries tend  to $L_1$ and $L_2$ respectively. Thus, as mentioned previously, only the $(1,1), (1,2), (2,1)$ and $(2,2)$ entries in $V_{\mathrm{d}}$, defined in \eqref{eqn:Vd}, are required. Note that $V_{\mathrm{c}}$ is a sub-matrix of $V_{\mathrm{d}}$ (in the $[1:2, 1:2]$ position).   Furthermore, by the fact that $V_{\mathrm{c}}$ is diagonal, the relation in~\eqref{eqn:diag} also holds. The rest of the arguments are similar to  case $1$.\\

\textbf{Case} $3$: When $\kappa_1 < I_{11}$ and $\kappa_2 = I_{21}$ \\
By symmetry, case $3$ is proved similarly to case $1$.

\subsection{Supporting lemmas}
 
This subsection contains  a few supporting lemmas, which will be used to prove the main result of this paper.

The following lemma gives  a variant of the multivariate Berry-Esseen Theorem \cite{Gotze91} \cite{BH10}, which is a restatement of Corollary~38 in \cite{WKT13}. The lemma can be applied to random vectors which are independent, but not necessarily identically distributed. For i.i.d.\ random vectors, interested readers can refer to Bentkus's work \cite{Bentkus2003}. This lemma is used in the converse proof of Theorem~\ref{T1}.
\begin{mylemma} \label{LemmaBE}
Let $\mathbf{U}_1,\ldots, \mathbf{U}_n$ be independent, zero-mean random vectors in $\mathbb{R}^m$. Let $\mathbf{G}_n \triangleq \frac{1}{\sqrt{n}}(\mathbf{U}_1+ \ldots +\mathbf{U}_n)$, $V \triangleq \mathrm{cov}(\mathbf{G}_n)$, $t \triangleq \frac{1}{n} \sum_{i=1}^n \mathbb{E}[\|\mathbf{U}_i\|_2^3]$ and let $\mathbf{Z} \sim \mathcal{N}(\mathbf{0},V) $. Let $\mathbb{C}_m$ be the family of all convex, Borel measurable subsets of $\mathbb{R}^m$. Assume $V  \succ 0$ and let the minimum eigenvalue of $V $ be $\lambda_{\min}(V)$. Then, for all $n \in \mathbb{N}$, we have
\begin{align}
\sup_{\mathfrak{C} \in \mathbb{C}_m} |\Pr(\mathbf{G}_n \in \mathfrak{C}) - \Pr(\mathbf{Z} \in \mathfrak{C})| \leq \frac{254\, \sqrt{m}\, t}{\lambda_{\min}(V )^{3/2}\sqrt{n}}.
\end{align}
\end{mylemma}

The following lemma provides  a variant of the multivariate Berry-Esseen Theorem \cite{Gotze91} \cite{BH10}, which is a restatement of Proposition 1 in \cite{ML13}. The lemma can be applied to functions of sums of i.i.d. random vectors under certain conditions. This lemma is used in the direct proof of Theorem \ref{T1}.

\begin{mylemma} \label{LemmaBE2}
Let $\{ \mathbf{U}_t \triangleq (U_{1t}, U_{2t}, \ldots, U_{at})\}_{t=1}^\infty$ be a sequence of zero-mean i.i.d.\ random vectors in $\mathbb{R}^a$ with  $\mathbb{E}[\|\mathbf{U}_t\|_2^3]$ being finite. Consider a vector-valued  function $\mathbf{g}  :\mathbb{R}^a \to \mathbb{R}^b$. Denote $\mathbf{g}(\mathbf{u} ) \triangleq [g_1(\mathbf{u} ), g_2(\mathbf{u} ), \ldots, g_{b}(\mathbf{u} )]^T$. Assume that ${\mathbf{g}}(\mathbf{u} )$ has continuous second-order partial derivatives in a  neighbourhood of $\mathbf{u}  = \mathbf{0}$ of side length at least ${n^{-1/4}}$. Denote the corresponding Jacobian matrix $J$ at $\mathbf{u}  = \mathbf{0}$ of ${\mathbf{g}}(\mathbf{u} )$ as $J \in \mathbb{R}^{b \times a}$, whose components are defined as
\begin{align}
J_{ji} \triangleq \frac{\partial {g}_j(\mathbf{u})}{\partial  u_i}\Bigg|_{\mathbf{u} = \mathbf{0}}
\end{align}
for $j \in \{1,2,\ldots,b \}$, and $i \in \{ 1,2,\ldots,a \}$. Let the random vector $\mathbf{Z}$ have distribution $\mathcal{N}( {\mathbf{g}}(\mathbf{0}),  \frac{1}{n}\,  J \, \mathrm{Cov}( \mathbf{U}_1)\,    J^T)$. Then, for any convex Borel-measurable set $\mathcal{D}$ in $\mathbb{R}^b$, there exists a finite positive constant $c$ such that
\begin{align}
\Bigg|\Pr \Bigg[  {\mathbf{g}} \Bigg( \frac{1}{n} \sum_{t=1}^n \mathbf{U}_t \Bigg) \in \mathcal{D} \Bigg] - \Pr[\mathbf{Z} \in \mathcal{D}] \Bigg| \leq \frac{c}{\sqrt{n}}
\end{align}
\end{mylemma}

%

\subsection{Proof of Lemma \ref{LemmaCalConverse}}  \label{SecCalC}
We have, for $k \in \{1,2, \ldots,n\}$,
\begin{align}
\tilde{i}_{11k}(x_{1k} x_{2k} Y_{1k}) &= \frac{1}{2} \log (h_{11}^2 P_1 +1) + \frac{(Y_{1k} - h_{21} x_{2k})^2}{2(1+h_{11}^2 P_1)} 
                             - \frac{(Y_{1k} - h_{11} x_{1k} - h_{21} x_{2k})^2}{2}. 
\end{align}
In this case, $\tilde{i}_{11k}(x_{1k} x_{2k} Y_{1k})$ has the same statistics as
\begin{align}
g_{11}(Z_{1k}) = \frac{1}{2} \log (h_{11}^2 P_1 +1) + \frac{(Z_{1k} + h_{11} x_{1k})^2}{2(1+ h_{11}^2 P_1)}
                          - \frac{Z_{1k}^2}{2}.
\end{align} 
Using this expression, we have
\begin{align}
\mathbb{E}[\tilde{i}_{11k}(x_{1k} x_{2k} Y_{1k})] &= \frac{1}{2} \log (h_{11}^2 P_1 +1) + \frac{1 + h_{11}^2 x_{1k}^2}{2(1+ h_{11}^2 P_1)}
                          - \frac{1}{2}, \\
\mathrm{var}[\tilde{i}_{11k}(x_{1k} x_{2k} Y_{1k})] &=  \frac{h_{11}^4 P_1^2 + 2h_{11}^2x_{1k}^2}{2(1+h_{11}^2  P_1)^2}.                         
\end{align}
Therefore,
\begin{align}
\mathbb{E} \left[\frac{1}{n}\sum_{k=1}^n  \tilde{i}_{11k}(x_{1k} x_{2k} Y_{1k}) \right] &= \frac{1}{2} \log (h_{11}^2 P_1 +1) + \frac{n + h_{11}^2  \| x_{1}^n \| ^2}{2n(1+ h_{11}^2 P_1)} - \frac{1}{2} \\
                 &= I_{11}.         
\end{align}
Next, we have
\begin{align}
\mathrm{var}\left[\frac{1}{\sqrt{n}}\sum_{k=1}^n \tilde{i}_{11k}(x_{1k} x_{2k} Y_{1k})\right] &\stackrel{(a)}= \frac{1}{n} \sum_{k=1}^n \mathrm{var} \left[ \tilde{i}_{11k}(x_{1k} x_{2k} Y_{1k})\right] \\
                                      &=\frac{1}{n} \cdot   \frac{n h_{11}^4 P_1^2 + 2h_{11}^2  \| x_{1}^n \| ^2}{2(1+h_{11}^2  P_1)^2} \\
                                      &= V_1,
\end{align}
where $(a)$ follows from the mutual independence of $Z_{1k}$'s.

Similarly,  $\tilde{i}_{21k}(x_{1k} x_{2k} Y_{2k})$  for $k \in \{1,2,...,n\}$  has the same statistics as
\begin{align}
g_{21}(Z_{2k}) = \frac{1}{2} \log (h_{22}^2 P_2 +1) + \frac{(Z_{2k} + h_{22} x_{2k})^2}{2(1+h_{22}^2 P_2)}
                          - \frac{Z_{2k}^2}{2},
\end{align} 
and its statistics are given by
\begin{align}
\mathbb{E}[\tilde{i}_{21}(x_{1k} x_{2k} Y_{2k})] &= \frac{1}{2} \log (h_{22}^2 P_2 +1) + \frac{1 + h_{22}^2 x_{2k}^2}{2(1+h_{22}^2 P_2)}
                          - \frac{1}{2}, \\
\mathrm{var}[\tilde{i}_{21}(x_{1k} x_{2k} Y_{2k})] &=  \frac{h_{22}^4 P_2^2 + 2h_{22}^2 x_{2k}^2}{2(1+h_{22}^2 P_2)^2}.                         
\end{align}

Similarly, we can find the mean and the variance of the sum of these information densities, yielding
\begin{align}
\mathbb{E}   \left[\frac{1}{n}\sum_{k=1}^n \tilde{\mathbf{i}}_{{\mathrm{c}}k}         ( x_{1k} x_{2k} Y_{1k} Y_{2k})\right] &= \mathbf{I}_{\mathrm{c}}, \\
 \mathrm{cov}\left[\frac{1}{\sqrt{n}}\sum_{k=1}^n \tilde{\mathbf{i}}_{{\mathrm{c}}k} ( x_{1k} x_{2k}  Y_{1k} Y_{2k}) \right] &= V_{\mathrm{c}}. \label{eqn:cov}
\end{align}
Interestingly, because $Z_{1j}$ is independent of $Z_{2k}$, we have
\begin{align}
\mathrm{cov} \big[ \tilde{i}_{11j}(x_{1j} x_{2j} Y_{1j}), \tilde{i}_{21k}(x_{1k} x_{2k} Y_{2k})\big] = 0,
\end{align}
for all $j,k \in \{1,2, \ldots, n\}$ with $j \neq k$. This leads directly to the diagonal covariance matrix in \eqref{eqn:cov}.
The lemma is proved.

\subsection{Proof of Lemma \ref{LemmaFiniteK} } \label{SecFiniteK}
Similar to  \cite[Lem.~61]{Polyanskiy2010}  and  \cite[Prop.~3]{ML13}, we can prove that $K_{11}$ and $K_{21}$ are upper bounded by a constant  when $n$ is sufficiently large.

The marginal conditional output distribution $P_{Y_1^n|X_2^n}$ induced by feeding the input distributions, given in (\ref{EDinput}), into the Gaussian IC can be shown to be 
\begin{align}
P_{Y_1^n|X_2^n}(y_1^n|x_2^n) &=\frac{1}{2 \pi^{n/2}} \Gamma \left(\frac{n}{2}\right) e^{-n h_{11}^2 P_1/2} e^{- \| y_1^n - h_{21} x_2^n \| ^2/2} 
                              \frac{I_{n/2-1}( \| y_1^n - h_{21} x_2^n \| \sqrt{nP_1} h_{11})}{( \| y_1^n - h_{21} x_2^n \| \sqrt{nP_1} h_{11})^{n/2-1}}, 
\end{align}
where $I_{v}(\cdot)$ is the modified Bessel function of the first kind and $v$-th order. The marginal distribution $P_{Y_2^n|X_1^n}$ has a similar form to the above.

We have
\begin{align}
D_{11}(y_1^n|x_2^n) &\triangleq \frac{P_{Y_1^n|X_2^n}(y_1^n|x_2^n)}{Q_{Y_1^n|X_2^n}(y_1^n|x_2^n)} \nonumber \\
               &=\frac{1}{2} \Gamma \left(\frac{n}{2} \right) [2 e^{-h_{11}^2P_1} (1+ h_{11}^2P_1)]^{n/2} e^{-\frac{h_{11}^2P_1 \| y_1^n - h_{21}x_2^n \| ^2}{2(1+h_{11}^2P_1)}} \cdot \frac{I_{n/2-1}( \| y_1^n -h_{21}x_2^n \| \sqrt{nP_1}h_{11})}{( \| y_1^n -h_{21}x_2^n \| \sqrt{nP_1}h_{11})^{n/2-1}}.
\end{align}

Note that the gamma function $\Gamma(\cdot)$ can take  different forms. Using Binet's first formula for $\log  \Gamma(z)$ \cite[Chap. ~1]{Erdelyi53V1}, we have
\begin{align}
\log  \Gamma(z) = \left(z-\frac{1}{2} \right) \log  z -z + \frac{1}{2} \log  (2 \pi) + \int_{0}^{\infty} \left(\frac{1}{2} - \frac{1}{t} + \frac{1}{e^t -1} \right) \frac{e^{-tz}}{t} \mathrm{d}t.
\end{align}
Note that the fourth term converges to $0$ as $z \to \infty$. Thus, we can upper-bound  $\Gamma \left(\frac{n}{2} \right)$ by
\begin{align}
\Gamma \left(\frac{n}{2} \right) \leq \left(\frac{n}{2}-\frac{1}{2} \right) \log  \frac{n}{2} -\frac{n}{2} + \frac{1}{2} \log  (2 \pi) + c_n
\end{align}
where $\{c_n\}_{n=1}^{\infty}$ is a sequence of numbers that converges to $0$.

From Prokhorov's work \cite{Prokhorov68} and~\cite[Lem.~61]{Polyanskiy2010}, when $k$ is even we can upper-bound the modified Bessel function as
\begin{align}
z^{-k} I_k(z) \leq \sqrt{\frac{\pi}{8}} (k^2 + z^2)^{-1/4} (k + \sqrt{k^2 + z^2})^{-k} e^{\sqrt{k^2+ z^2}}.
\end{align}
Note that $I_{n/2 -1}(\cdot) < I_{n/2 -3/2}(\cdot)$. When $n$ is odd, an upper bound is obtained by replacing $I_{n/2 -1}(\cdot)$ by $I_{n/2 -3/2}(\cdot)$. Thus, it is sufficient to consider the upper bound on $D(y_1^n|x_2^n)$ when $n$ is even.

After some manipulations, we can show that
\begin{align}
D_{11}(y_1^n|x_2^n) \leq \exp\left[ c_{11} + c_n + \frac{n}{2}\phi_{\xi,P_1,n}\left( \frac{ \| y_1^n - h_{21}x_2^n \| ^2}{n} \right) \right],
\end{align}
where
\begin{align}
c_{11} &\triangleq \log  \frac{1}{2} + \log  \sqrt{\frac{{\pi}}{8}} + \frac{1}{2} \log  (2 \pi) \\
\phi_{\xi,P_1,n}(z) &\triangleq \log  \left(2(1+h_{11}^2P_1) e^{-(1+h_{11}^2 P_1)}\right) - \frac{h_{11}^2 P_1 z}{h_{11}^2 P_1 +1} + \sqrt{\xi^2 + 4 h_{11}^2 P_1 z} \nonumber\\*
 & \qquad \qquad - \xi \log  \left(\xi + \sqrt{\xi^2 + 4 h_{11}^2 P_1 z} \right)- \frac{1- \xi}{2} \log  \left(\sqrt{\xi^2 + 4 h_{11}^2 P_1 z} \right) \\
\xi  &\triangleq \frac{n/2 -1}{n/2}.       
\end{align}

Note that 
\begin{align}
\lim_{n \to \infty} \phi_{\xi,P_1,n}(z) = \phi_{P_1} (z),
\end{align}
where
\begin{align}
\phi_{P_1}(z) &\triangleq \log  \left(2(1+h_{11}^2P_1) e^{-(1+h_{11}^2 P_1)}\right) - \frac{h_{11}^2 P_1 z}{h_{11}^2 P_1 +1} + \sqrt{1 + 4 h_{11}^2 P_1 z} -  \log  \left(1 + \sqrt{1 + 4 h_{11}^2 P_1 z} \right). 
\end{align}
It can be shown that $\phi_{P_1}(z) \leq 0$. Equality occurs when $z = 1+h_{11}^2 P_1$. Therefore, we have $K_{11}$ is   upper bounded by a constant, when $n$ is sufficiently large.  Similarly, we can shown that $K_{21}$ is   upper bounded by a constant when $n$ is sufficiently large.

It is hard to derive a closed-form expression for the output distribution $P_{Y_1^n}$ induced by the input distributions in (\ref{EDinput}) and the IC. However, we can characterize the distribution of $B^n \triangleq h_{11} X_1^n + h_{21} X_2^n$ (see \cite[Equations (137-151)]{ML13}). We have
\begin{align}
P_{B^n}(b^n) = 
\begin{cases}
&0 \quad \text{   if $ \| b^n \|  \leq |h_{11} \sqrt{nP_1} - h_{21}\sqrt{nP_2}| $} \\
&0 \quad \text{   if $ \| b^n \|  \geq |h_{11} \sqrt{nP_1} + h_{21}\sqrt{nP_2}| $} \\
&\phi_\mathrm{B}(b^n) \quad \text{   otherwise},\\
\end{cases}
\end{align}
where
\begin{align}
\phi_\mathrm{B}(b^n) &\triangleq \frac{1}{h_{21}^n} \sqrt{\frac{P_2}{\pi P_1}}\frac{h_{21}}{h_{11}}\frac{\Gamma(\frac{n}{2})}{\Gamma(\frac{n-1}{2})} \frac{1}{S_n(\sqrt{nP_2})} \frac{1}{ \| b^n \| } \left(1 - \left(\frac{ \| b^n \| ^2 + n(h_{11}^2P_1 - h_{21}^2 P_2)}{2h_{11}\sqrt{nP_1}  \| b^n \| }\right)^2\right)^{(n-3)/2} \\
\cos \theta_0 &\triangleq \frac{ \| b^n \| ^2 + n(h_{11}^2P_1 - h_{21}^2 P_2)}{2 h_{11} \sqrt{n P_1}  \| b^n \| }. 
\end{align}

 Define the auxiliary input distribution $Q_{B^n}(b^n) \triangleq N(b^n;\mathbf{0},(h_{11}^2 P_1 + h_{21}^2 P_2)\mathbf{I}_{n\times n})$. If this distribution is used as an input for the channel $Y_1^n = B^n + Z_1^n$, the corresponding output distribution is $Q_{Y_1^n}$. If it can be proved that 
\begin{equation}
 K_{12}^{\prime} \triangleq\sup_{b^n} \frac{P_{B^n}(b^n)}{Q_{B^n}(b^n)}
 \end{equation} 
 is uniformly bounded when $n$ is sufficiently large, then, for any $y_1^n$, we have
\begin{align}
P_{Y_1^n} (y_1^n) &= \int_{\mathbb{R}^n} P_{B^n}(b^n) P_{Y_1^n|B^n}(y_1^n|b^n) \mathrm{d}b^n \nonumber \\
          &\leq \int_{\mathbb{R}^n} K_{12}^{\prime} Q_{B^n}(b^n) P_{Y_1^n|B^n}(y_1^n|b^n) \mathrm{d}b^n \nonumber \\
          &=  K_{12}^{\prime} Q_{Y_1^n} (y_1^n).
\end{align}
Therefore, $K_{12} \leq K_{12}^{\prime}$. That is, $K_{12}$ is uniformly bounded when    $n$ is sufficiently large. Now, we   prove the finiteness of $K_{12}^{\prime}$. Define 
\begin{align}
D_{12} (b^n) \triangleq \frac{P_{B^n}(b^n)}{Q_{B^n}(b^n)}.
\end{align}

Next, by simple algebraic manipulations, it can be shown that
\begin{align}
D_{12} (b^n) \leq \exp \left[ c_{12} + c_n+ \rho_{12n}\left( \frac{ \| b^n \| ^2}{n}  \right) \right] 
\end{align}
where
\begin{align}
c_{12} &\triangleq   \log  \left( \frac{P_2}{\sqrt{\pi P_1}} \frac{h_{21}}{h_{11}} \right)  + \frac{\log (2\pi)}{2} \\
\rho_{12n}(z) &\triangleq -\frac{\log  z}{n} + \log  \frac{h_{11}^2P_1 + h_{21}^2 P_2}{ e h_{21}^2 P_2} + \frac{z}{h_{11}^2P_1 + h_{21}^2 P_2} + \frac{n-3}{n} \log  \left(1 - \frac{( z+ h_{11}^2 P_1 - h_{21}^2 P_2)^2}{4 h_{11}^2 P_1 z} \right),
\end{align}
and where $\{c_n\}$ is a sequence converging to $0$, and $|h_{11} \sqrt{nP_1} - h_{21}\sqrt{nP_2}| < z < |h_{11} \sqrt{nP_1} + h_{21}\sqrt{nP_2}|$. 

Note that 
\begin{align}
\lim_{n \to \infty} \rho_{12n}(z) = \rho_{12}(z),
\end{align}
where
\begin{align}
\rho_{12}(z) &\triangleq  \log  \frac{h_{11}^2P_1 + h_{21}^2 P_2}{ e h_{21}^2 P_2} + \frac{z}{h_{11}^2P_1 + h_{21}^2 P_2} + \log  \left(1 - \frac{( z+ h_{11}^2 P_1 - h_{21}^2 P_2)^2}{4 h_{11}^2 P_1 z} \right).
\end{align}
It can be shown that $\rho_{12}(z) \leq 0$. Equality occurs at $z= h_{11}^2P_1 + h_{21}^2 P_2$.
 Thus, we can conclude that $K_{12}^{\prime}$ is   upper bounded by a constant when $n$ is sufficiently large.  Similarly, $K_{22}$ can be proved to be upper bounded by a constant for $n$ sufficiently large.

\subsection{Proof of Lemma \ref{LemmaVH}} \label{Sec:ProofVH}
Given the joint distribution in (\ref{Ed}), denote the marginal distributions and the conditional distributions of this distribution  as $P_{Y_1^n X_1^n X_2^n }(y_1^n x_1^n x_2^n )$, $P_{Y_2^n X_1^n X_2^n }(y_2^n x_1^n x_2^n)$,  $P_{X_1^n|X_2^n }(x_1^n|x_2^n)$, and $P_{X_2^n|X_1^n }(x_2^n|x_1^n)$, where
\begin{align}
P_{Y_1^n X_1^n X_2^n }(y_1^n x_1^n x_2^n ) &\triangleq \sum_{y_2^n} P_{Y_2^n Y_1^n X_1^n X_2^n }(y_2^n y_1^n x_1^n x_2^n ), \label{Emarginal1}\\
P_{X_1^n|X_2^n }(x_1^n|x_2^n)              &\triangleq \frac{\sum_{y_2^n y_1^n} P_{Y_2^n Y_1^n X_1^n X_2^n }(y_2^n y_1^n x_1^n x_2^n )}{P_{X_2^n }(x_2^n)} \label{Emarginal2},
\end{align}
and the remaining distributions are defined similarly.

Define the decoding regions
\begin{align}
D_{1s_1} &\triangleq \{y_1^n \in \mathcal{Y}_1^n | g_{1n}(y_1^n) = s_1 \} \\
D_{2s_2} &\triangleq \{y_2^n \in \mathcal{Y}_2^n | g_{2n}(y_2^n) = s_2 \}  \\
D_{1s_1}^\prime &\triangleq \{(y_1^n y_2^n) \in \mathcal{Y}_1^n \times \mathcal{Y}_2^n | y_1^n \in D_{1s_1} \} \\
D_{2s_2}^\prime &\triangleq \{(y_1^n y_2^n) \in \mathcal{Y}_1^n \times \mathcal{Y}_2^n | y_2^n \in D_{2s_2} \},
\end{align}
where $s_1 \in \{1,2,\ldots,M_{1n} \}$ and $s_2 \in \{1,2,\ldots,M_{2n} \}$.\\

The decoding functions $g_{jn}$ and the encoding functions $f_{jn}$, for $j=1,2,$  in this proof, are defined in the section for problem formulation.

Note that
\begin{align}
\frac{ W_{1}^n (y_{1}^n|x_{1}^n x_{2}^n)}{Q_{Y_1^n|X_2^n}(y_1^n|x_2^n)} 
&= \frac{ P_{Y_1^n X_1^n X_2^n} (y_{1}^n x_{1}^n x_{2}^n)}{Q_{Y_1^n X_2^n}(y_1^n x_2^n) P_{X_1^n|X_2^n}(x_1^n|x_2^n)} \\
&\stackrel{(a)}= \frac{ P_{Y_1^n X_1^n X_2^n} (y_{1}^n x_{1}^n x_{2}^n)}{Q_{Y_1^n X_2^n}(y_1^n x_2^n) P_{X_1^n}(x_1^n)} \\
&\stackrel{(b)}= M_{1n} \frac{ P_{Y_1^n X_1^n X_2^n} (y_{1}^n x_{1}^n x_{2}^n)}{Q_{Y_1^n X_2^n}(y_1^n x_2^n) },
\end{align}
where
\begin{enumerate}[(a)]
\item follows from the fact that $X_1^n$ and $X_2^n$ are independent; and

\item follows from the fact that $P_{X^n_1}(x_1^n) = \frac{1}{M_{1n}}$ for all $x_1^n$ in the first codebook.
\end{enumerate}  

Similarly, we have
\begin{align}
\frac{ W_{2}^n (y_{2}^n|x_{1}^n x_{2}^n)}{Q_{Y_2^n|X_1^n}(y_2^n|x_1^n)} 
= M_{2n} \frac{ P_{Y_2^n X_1^n X_2^n} (y_{2}^n x_{1}^n x_{2}^n)}{Q_{Y_2^n X_1^n}(y_2^n x_1^n) }.
\end{align}
Define
\begin{align}
B_{1s_1s_2} &\triangleq \Big\{ y_1^n \in \mathcal{Y}_1^n \Big| \frac{ P_{Y_1^n X_1^n X_2^n} (y_{1}^n f_{1n}(s_1) f_{2n}(s_2))}{Q_{Y_1^n X_2^n}(y_1^n f_{2n}(s_2))} \leq e^{-n\gamma} \Big\} \\
B_{1s_1s_2}^\prime &\triangleq  \{(y_1^n y_2^n) \in \mathcal{Y}_1^n \times \mathcal{Y}_2^n| y_1^n \in B_{1s_1s_2} \} \\
B_{2s_1s_2} &\triangleq \Big\{ y_2^n \in \mathcal{Y}_2^n \Big| \frac{ P_{Y_2^n X_1^n X_2^n} (y_{2}^n f_{1n}(s_1) f_{2n}(s_2))}{Q_{Y_2^n X_1^n}(y_2^n f_{1n}(s_1))} \leq e^{-n\gamma} \Big\} \\
B_{2s_1s_2}^\prime &\triangleq  \{(y_1^n y_2^n) \in \mathcal{Y}_1^n \times \mathcal{Y}_2^n| y_2^n \in B_{2s_1s_2} \},
\end{align}
where $s_1 \in \{1,2,\ldots,M_{1n} \}$ and $s_2 \in \{1,2,\ldots,M_{2n} \}$.\\

Define 
\begin{align}
G_{1} &\triangleq \Big\{(x_1^n x_2^n y_1^n y_2^n) \in \mathcal{X}_1^n \times \mathcal{X}_2^n \times \mathcal{Y}_1^n \times \mathcal{Y}_2^n  \Big| \frac{ P_{Y_1^n X_1^n X_2^n} (y_{1}^n x_{1}^n x_{2}^n)}{Q_{Y_1^n X_2^n}(y_1^n x_2^n)} \leq e^{-n\gamma} \Big\} \\
G_{2} &\triangleq \Big\{(x_1^n x_2^n y_2^n y_2^n) \in \mathcal{X}_1^n \times \mathcal{X}_2^n \times \mathcal{Y}_1^n \times \mathcal{Y}_2^n \Big| \frac{ P_{Y_2^n X_1^n X_2^n} (y_{2}^n x_{1}^n x_{2}^n)}{Q_{Y_2^n X_1^n}(y_2^n x_1^n)} \leq e^{-n\gamma} \Big\},
\end{align}
where $s_1 \in \{1,2,\ldots,M_{1n} \}$ and $s_2 \in \{1,2,\ldots,M_{2n} \}$.\\

In order to prove this lemma, it suffices to prove
\begin{align}
P_{X_1^nX_2^n Y_1^n Y_2^n}(G_1 \cup G_2) \leq \epsilon_n + 2 e^{-n\gamma}.
\end{align}
We are going to prove the validity of this inequality. We have
\begin{align}
P_{X_1^nX_2^n Y_1^n Y_2^n}(G_1 \cup G_2) 
&= \sum_{s_1=1}^{M_{1n}} \sum_{s_2=1}^{M_{2n}} P_{X_1^nX_2^n Y_1^n Y_2^n} (f_{1n}(s_1) f_{2n}(s_2), B_{1s_1s_2}^\prime \cup B_{2s_1s_2}^\prime) \\
&= \sum_{s_1=1}^{M_{1n}} \sum_{s_2=1}^{M_{2n}} [P_{X_1^nX_2^n Y_1^n Y_2^n} (f_{1n}(s_1) f_{2n}(s_2),
 (B_{1s_1s_2}^\prime \cup B_{2s_1s_2}^\prime) \cap (D_{1s_1} \times D_{2s_2})^c) \nonumber \\
&\quad + P_{X_1^nX_2^n Y_1^n Y_2^n} (f_{1n}(s_1) f_{2n}(s_2),(B_{1s_1s_2}^\prime \cup B_{2s_1s_2}^\prime) \cap (D_{1s_1} \times D_{2s_2})) ] \\
&\leq \sum_{s_1=1}^{M_{1n}} \sum_{s_2=1}^{M_{2n}} [P_{X_1^nX_2^n Y_1^n Y_2^n} (f_{1n}(s_1) f_{2n}(s_2),  (D_{1s_1} \times D_{2s_2})^c) \nonumber \\*
&\quad + P_{X_1^nX_2^n Y_1^n Y_2^n} (f_{1n}(s_1) f_{2n}(s_2) ,(B_{1s_1s_2}^\prime \cup B_{2s_1s_2}^\prime) \cap (D_{1s_1} \times D_{2s_2})) ] \\
&\leq \epsilon_n + \sum_{s_1=1}^{M_{1n}} \sum_{s_2=1}^{M_{2n}} [ P_{X_1^nX_2^n Y_1^n Y_2^n} (f_{1n}(s_1) f_{2n}(s_2) ,B_{1s_1s_2}^\prime  \cap (D_{1s_1} \times D_{2s_2})) \nonumber \\*
&\quad + P_{X_1^nX_2^n Y_1^n Y_2^n} (f_{1n}(s_1) f_{2n} (s_2), B_{2s_1s_2}^\prime \cap (D_{1s_1} \times D_{2s_2})) ].
 \end{align}
Next, we upper-bound the second and third terms. We have
\begin{align}
&\sum_{s_1=1}^{M_{1n}} \sum_{s_2=1}^{M_{2n}}  P_{X_1^nX_2^n Y_1^n Y_2^n} (f_{1n}(s_1)f_{2n}(s_2), B_{1s_1s_2}^\prime  \cap (D_{1s_1} \times D_{2s_2})) \\
&\leq \sum_{s_1=1}^{M_{1n}} \sum_{s_2=1}^{M_{2n}}  P_{X_1^nX_2^n Y_1^n Y_2^n} (f_{1n}(s_1) f_{2n}(s_2), B_{1s_1s_2}^\prime  \cap D_{1s_1}^\prime ) \\
&=\sum_{s_1=1}^{M_{1n}} \sum_{s_2=1}^{M_{2n}} \sum_{(y_1^n y_2^n) \in B_{1s_1s_2}^\prime  \cap D_{1s_1}^\prime }   P_{X_1^nX_2^n Y_1^n Y_2^n} (f_{1n}(s_1) f_{2n}(s_2)y_1^n y_2^n) \\
&=\sum_{s_1=1}^{M_{1n}} \sum_{s_2=1}^{M_{2n}} \sum_{y_1^n  \in B_{1s_1s_2}  \cap D_{1s_1}}   P_{X_1^nX_2^n Y_1^n } (f_{1n}(s_1) f_{2n}(s_2) y_1^n) \\
&\stackrel{(a)}\leq \sum_{s_1=1}^{M_{1n}} \sum_{s_2=1}^{M_{2n}} \sum_{y_1^n  \in B_{1s_1s_2}  \cap D_{1s_1}}   Q_{X_2^n Y_1^n } ( f_{2n}(s_2) y_1^n) e^{-n\gamma} \\
&\leq \sum_{s_1=1}^{M_{1n}} \sum_{s_2=1}^{M_{2n}} \sum_{y_1^n  \in D_{1s_1}}   Q_{X_2^n Y_1^n } ( f_{2n}(s_2) y_1^n) e^{-n\gamma} \\ 
&= \sum_{s_2=1}^{M_{2n}}   Q_{X_2^n } ( f_{2n}(s_2) ) e^{-n\gamma} \\*
&\leq  e^{-n\gamma}, 
\end{align}
where (a) follows from the definition of $B_{1s_1s_2}$.\\

Similarly to the above, we can show that
\begin{align}
\sum_{s_1=1}^{M_{1n}} \sum_{s_2=1}^{M_{2n}} P_{X_1^nX_2^n Y_1^n Y_2^n} (f_{1n}(s_1) f_{2n}(s_2),  B_{2s_1s_2}^\prime \cap (D_{1s_1} \times D_{2s_2}))
\leq e^{-n\gamma}.
\end{align}
Thus, we have proved the lemma.

\subsection{Proof of Lemma \ref{LemmaF}} \label{Sec:ProofF}
First, we consider the case without cost constraints. Define the sets
\begin{align}
T_{j1} &\triangleq \left\{(x_1^n x_2^n y_j^n) \in \mathcal{X}^n_1 \times \mathcal{X}_2^n \times \mathcal{Y}_j^n | \tilde{i}^n_{j1} > \log M_{jn} + n\gamma  \right\} \\
T_{j2} &\triangleq \left\{(x_1^n x_2^n y_j^n) \in \mathcal{X}^n_1 \times \mathcal{X}_2^n \times \mathcal{Y}_j^n | \tilde{i}^n_{j2} > \log M_{1n} M_{2n} + n\gamma  \right\} \\
T_j    &= T_{j1} \cap T_{j2},
\end{align}
where the modified information densities $\tilde{i}^n_{j1}$ and $\tilde{i}^n_{j2}$ are defined in (\ref{eqn:i11_tilde}) and (\ref{eqn:i12_tilde}).

a) \textit{Codebook generation}\\
Fix a joint distribution $P_{X_1^n}(x_1^n) P_{X_2^n}(x_2^n)$. Generate $M_{jn}$ codewords $f_{jn}(s_j)$, for $ s_{j} \in \{ 1,2,...,M_{jn}\}$, and $j=1,2$. We denote the random codewords $f_{jn}(s_j)$ as $X_j^n(s_j)$ in the proof of this lemma.
\\

b) \textit{Encoding rules at transmitters}:\\
To transmit message $s_j$, transmitter $j$ sends the codewords $X_j^n(s_j)$.
\\

c) \textit{Decoding rules at receivers} \\
Upon receiving an output $y_1^n$, receiver $1$ finds the unique message $\hat{s_1}$ such that 
\begin{align}
(x_1^n(\hat{s_1}) x_2^n(\hat{s}_2) y_1^n) \in T_1^n \label{Ec}
\end{align}
for some $\hat{s_2}$. An error is declared otherwise. This decoding rule is also known as {\em simultaneous non-unique decoding rule} \cite[Section 6.2]{elgamal}.
The decoding rule at receiver $2$ is  defined similarly to the above. \\

d) \textit{Calculation of probability of error}\\
For ease of presentation, we define the event, for $j=1,2$,
\begin{align}
E_{js_1 s_2} \triangleq \{ ( (X_1^n(s_1) X_{2}^n(s_2)Y_j^n) \in T_j^n \}.
\end{align}

Decoding errors at receiver $1$ is bounded as
\begin{align}
          &\frac{1}{M_{1n}M_{2n}} \sum_{s_1=1}^{M_{1n}} \sum_{s_2=1}^{M_{2n}} 
\Bigg[ \Pr(E_{1s_1s_2}^c) +\Pr\Bigg(\bigcup_{s_1^\prime \neq s_1, \text{ any } s_2\prime} E_{1s_1^\prime s_2^\prime}  \Bigg)  \Bigg] \\
           &\stackrel{(a)} =   \Pr(E_{111}^c)  + 
              \Pr\left(\bigcup_{s_1^\prime \neq 1, \text{ any } s_2\prime} E_{1s_1^\prime s_2^\prime}  \right)   \\ 
           &\stackrel{(b)} \leq  \Pr(E_{111}^c) + \sum_{s_1^\prime \neq 1} \Pr (E_{1s_1^\prime 1}) 
+ \sum_{s_1^\prime \neq 1, s_2^\prime \neq 1} \Pr (E_{1s_1^\prime s_2^\prime}),
\end{align}
where 
\begin{enumerate}[(a)]
\item follows from the symmetry of the codebooks, and

\item follows from the union rule.
\end{enumerate}

Next, we bound the second term in the equation right above. \\
\begin{align}
\sum_{s_1^\prime \neq 1} \Pr (E_{1 s_1^\prime 1}) 
&= (M_{1n} -1) \Pr (\{(X_1^n(s_1^\prime) X_2^n(1) Y_1^n) \in T_{1} \}) \\
&\stackrel{(a)}= (M_{1n} -1) \sum_{(x_1^n x_2^n y_1^n) \in T_{1}} P_{X_1^n}(x_1^n) P_{X_2^n Y_1^n}(x_2^n y_1^n) \\
&\leq (M_{1n} -1) \sum_{(x_1^n x_2^n  y_1^n) \in T_{11}} P_{X_1^n}(x_1^n) P_{X_2^n Y_1^n}(x_2^n y_1^n) \\
&\leq (M_{1n} -1) \sum_{(x_1^n x_2^n  y_1^n) \in T_{11}} K_{11} P_{X_1^n}(x_1^n) Q_{X_2^n Y_1^n}(x_2^n y_1^n) \\
&\stackrel{(b)}\leq (M_{1n} -1) \sum_{(x_1^n  x_2^n y_1^n) \in T_{11}}  K_{11} P_{X_1^n}(x_1^n) 
                P_{X_2^n}(x_2^n) W_1( y_1^n|x_2^n x_1^n) e^{-n \gamma} \frac{1}{M_{1n}}  \\
&\leq   K_{11}  e^{-n \gamma}  
\end{align}
where
\begin{enumerate}[(a)]
\item follows from the fact that $X_{1}^n(s_1^\prime)$ and $(X_{2}^n(1) Y_1^n)$ are independent, when message pair $(1,1)$ are transmitted by transmitters, and
\item follows from the definition of the set $T_{11}$.
\end{enumerate}
Similarly, we can show that 
\begin{align}
\sum_{s_1^\prime \neq 1} \Pr (E_{1 s_1^\prime s_2^\prime}) 
\leq   K_{12}  e^{-n \gamma}. 
\end{align}

Similarly, we can upper-bound the decoding error events at receiver $2$ by 
\begin{align}
      &\frac{1}{M_{1n} M_{2n}} \sum_{s_1=1}^{M_{1n}} \sum_{s_2=1}^{M_{2n}} \Bigg[ \Pr(E_{2s_1s_2}^c)  + \Pr \Bigg(\bigcup_{s_2^\prime \neq s_2, \text{ any } s_1\prime} E_{2s_1^\prime s_2^\prime}  \Bigg)  \Bigg] \\
&\leq\Pr(E_{211}^c) + (K_{21}+ K_{22}) e^{-n\gamma}.
\end{align}

Therefore, we have
\begin{align}
\epsilon_n &\leq \Pr(E_{111}^c \cup E_{211}^c)   +  (K_{11}+ K_{12}) e^{-n\gamma}  + (K_{21}+ K_{22}) e^{-n\gamma}  \\
         &= \Pr (\mathcal{E}_{11} \cup \mathcal{E}_{12} \cup \mathcal{E}_{21} \cup \mathcal{E}_{22}) + K e ^{-n\gamma}. 
\end{align}

In the case where the cost constraint is imposed, we have
\begin{align}
\epsilon_n 
          &\leq \Pr (\mathcal{E}_{11} \cup \mathcal{E}_{12} \cup \mathcal{E}_{21} \cup \mathcal{E}_{22}) + K e ^{-n\gamma} 
         + P_{X_1^n} P_{X_2^n}(\{X_1^n \not \in \mathcal{F}_{1n} \cup X_2^n \not \in \mathcal{F}_{2n}\}).
\end{align}

Thus, we have proved the lemma.

\subsubsection*{Acknowledgments}
The authors would like to acknowledge several helpful discussions with Jonathan Scarlett and Masahito Hayashi. 

\bibliographystyle{unsrt}
\bibliography{myrefOPT}
\end{document}